\documentclass[10pt, twoside, twocolumn]{IEEEtran}

\usepackage[latin1]{inputenc}
\usepackage[T1]{fontenc}
\usepackage{amssymb,mathtools}   

\usepackage{amsthm}
\usepackage{graphicx}
\usepackage{subfigure}
\usepackage{psfrag}

\usepackage{color}
\usepackage{url}
\usepackage{hyperref}

\usepackage{adjustbox}
\usepackage{suffix}
\usepackage{booktabs}  % tabla
\usepackage{multirow} %tabla
\usepackage{relsize}
\usepackage{array}
\newtheorem{prop}{Proposition}

\newtheorem{theorem}{Theorem}
\newtheorem{lemma}[theorem]{Lemma}

\usepackage[numbers,sort&compress]{natbib}

\DeclarePairedDelimiter\abs{\lvert}{\rvert} %para abs
\usepackage{float}
\usepackage{stfloats}
\usepackage{textcomp}

\begin{document}

\title{The {M}ulti-{C}luster {F}luctuating \\{T}wo-{R}ay {F}ading {M}odel}

\author{Jos\'e~David~Vega-S\'anchez, F. Javier L\'opez-Mart\'inez, Jos\'e F. Paris and Juan M. Romero-Jerez
\thanks{Manuscript received MONTH xx, YEAR; revised XXX. The review of this paper was coordinated by XXXX. The work of F.~J.~L\'opez-Mart\'inez, J.~F.~Paris and J.~M.~Romero-Jerez was funded in part by Junta de Andaluc\'ia, the European Union and the European Fund for Regional Development FEDER through grants P21-00420, P18-RT-3175 and EMERGIA20-00297, and in part by MCIN/AEI/10.13039/501100011033 through grant PID2020-118139RB-I00. (\textit{Corresponding author: Jos\'e~David~Vega-S\'anchez})}
\thanks{J.~D.~Vega-S\'anchez is with the Faculty of Engineering and Applied Sciences, Telecommunications Engineering, Universidad de Las Am\'ericas (UDLA), 170503 Quito, Ecuador. (E-mail: $\rm jose.vega.sanchez@udla.edu.ec$).}
%,and also with the Departamento de Electr\'onica, Telecomunicaciones y Redes de Informaci\'on, Escuela Polit\'ecnica Nacional (EPN), Quito,  170525, Ecuador. (E-mail: $\rm jose.vega01@epn.edu.ec$)}
\thanks{F.~J.~L\'opez-Mart\'inez, J.~F.~Paris and J.~M.~Romero-Jerez are with the Communications and Signal Processing Lab, Telecommunication Research Institute (TELMA), Universidad de M\'alaga, M\'alaga, 29010, (Spain). F.~J.~L\'opez-Mart\'inez is also with the Dept. Signal Theory, Networking and Communications, University of Granada, 18071, Granada (Spain). (E-mails: $\rm fjlm@ugr.es, paris@ic.uma.es, romero@dte.uma.es$)}
}

\maketitle
%\(\)\thispagestyle{empty}
%\pagestyle{empty}

%% ABSTRACT
\begin{abstract}
We introduce and characterize the {Multi-cluster Fluctuating Two-Ray (MFTR)} fading channel, generalizing \emph{both} the fluctuating two-ray (FTR) and the \mbox{$\kappa$-$\mu$} shadowed fading models through a more general yet equally {\mbox{mathematically}} tractable model. We derive all the chief probability functions of the MFTR model such as probability density function (PDF), cumulative {distribution} function (CDF), and moment {generating} function (MGF) in closed-form, having {a mathematical complexity similar to} other fading models in the state-of-the-art. We also provide two additional analytical formulations for the PDF and the CDF: (\emph{i}) in terms of a continuous mixture of $\kappa$-$\mu$ shadowed distributions, and (\emph{ii}) as an infinite {discrete} mixture of Gamma distributions. Such expressions enable to conduct performance analysis under MFTR fading by \emph{directly} leveraging readily available results for the $\kappa$-$\mu$ shadowed or Nakagami-$m$ cases, respectively. {We demonstrate that the MFTR fading model provides a much better fit than FTR and $\kappa$-$\mu$ shadowed models for small-scale measurements of channel amplitude in outdoor Terahertz (THz) wireless links.} Finally, the performance of wireless {communications} systems undergoing MFTR fading is exemplified in terms of classical benchmarking metrics like the outage probability, both in exact and asymptotic forms, and the amount of fading. %To the best of the author's current knowledge, the MFTR model is likely the most generalized and unifying fading channel introduced in the literature. Therefore, a large number of new issues will arise with this channel model, and they constitute a prolific research field. Finally, the validity of the MFTR expressions is confirmed via Monte Carlo simulations. 
\end{abstract}


%%%%%%%%%%%%%%%%%%%%%%%%%%%%%%%%%%%%%%%%%%%%%%% INTRODUCTION

\section{Introduction}
\IEEEPARstart{A}{s} radio signals travel along a certain propagation environment, their interaction with objects, system agents and the propagation medium itself causes a number of effects that ultimately determine the amount of power being received at a given location. Classically, these phenomena have been well-characterized through the central limit theorem (CLT), that gives rise to a family of Gaussian models, being the most popular and widespread-used the Rayleigh/Rician models and also including more general ones like the Hoyt and Beckmann models \cite{Simon,Beckmann1962}. While CLT-based models have sufficed to characterize the effects of fading for decades, the rapid growth of wireless technologies along the $21^{\rm st}$ century requires more sophisticated efforts to properly capture the intrinsic nature of wireless propagation as we move up in frequency (where the CLT assumption may no longer hold) \cite{ftr,Marins2019}, or when new use cases demand for improved accuracy, especially in {highly} low-outage regimes \cite{Eggers2019,Mehrnia2022}. 
%\vspace{5mm}

The key advances in wireless fading channel modeling during the last years have been largely based on two different approaches. On the one hand, power envelope-based formulations were used by Yacoub to propose two families of generalized models: the $\kappa$-$\mu$ and the $\eta$-$\mu$ distributions \cite{kappamuYacub}. These distributions were later generalized and unified through the $\kappa$-$\mu$ shadowed fading distribution \cite{kms,Laureano2016}, which has become rather popular in the literature thanks to its ability to model a wide number of propagation conditions through a limited set of physically-justified parameters. On the other hand, ray-based formulations express the received signal as a superposition of a number of scattered waves with random phases \cite{durgin2000theory}. Despite their clear physical motivation, the mathematical complexity associated to ray-based models grows when more than a few dominant waves are individually accounted for \cite{Romero2022}. However, Durgin's Two-Wave with Diffuse Power (TWDP) fading model and its subsequent generalizations \cite{durgin,Rao2015,ftr} have also managed to become popular in the context of wireless channel modeling for higher-frequency bands. Specifically, the Fluctuating Two-Ray (FTR) fading model \cite{ftr} has {widely} been adopted in mm-Wave  environments, and even recently in terahertz (THz) bands \cite{Du2022}.

Key to their success, the $\kappa$-$\mu$ shadowed and the FTR fading models share a number of features that are desirable for a stochastic fading model to be of practical use: (\textit{i}) have a small number (three, in both cases) of physically justified parameters; (\textit{ii}) have a reasonably good mathematical tractability; (\textit{iii}) include simpler but popular fading models as special cases. However, because these two models arise from different physical formulations, they do not capture the same type of propagation behaviors. For instance, the FTR model has the ability to exhibit bimodality that often appears in field measurements \cite{Mavridis15,Yoo16,Du2022}, while the $\kappa$-$\mu$ shadowed one is inherently unimodal in its original formulation \cite{kms}. On the other hand, the $\kappa$-$\mu$ shadowed model can exhibit any sort of power-law asymptotic decay (diversity order), while ray-based models asymptotically decay with unitary slope \cite{Romero2022}.

Aiming to reconcile the two dominant approaches in the literature of stochastic fading channel modeling, which have hitherto evolved separately, we here introduce the Multi-cluster Fluctuating Two-Ray (MFTR) fading model. This newly proposed model is presented as the natural generalization and unification of \emph{both} the FTR and the \mbox{$\kappa$-$\mu$} shadowed fading models.
%\comment{The fluctuating two-ray (FTR) \cite{ftr} and the \mbox{$\kappa$-$\mu$%} shadowed \cite{kms} fading  models have been recently proposed and enthusiastically, adopted by the wireless communications scientific community with hundreds of citations in Google Scholar. Each of these models focus on a different direction to generalize the classical fading models (Rayleigh, Rician, Nakagami-$m$ and Nakagami-$q$ (Hoyt)). While the FTR model is based on generalizing the physical structure of the multiple ray models, the $\kappa$-$\mu$ relies on using multiple clusters to generalize the Rician shadowed model. Within this context, a new fading channel model -- the Multi-cluster Fluctuating Two-Ray (MFTR) -- is proposed in this work as the natural unification of both, the FTR and the \mbox{$\kappa$-$\mu$} fading models.
Such generalization enables the presence of additional multipath clusters in the purely ray-based FTR model, and the consideration of two fluctuating specular components within {one of the clusters, typically the first one received, of the} \mbox{$\kappa$-$\mu$} shadowed formulation. Despite being a generalization of the FTR and the \mbox{$\kappa$-$\mu$} shadowed models, only one additional parameter is required over these baseline models. As we will later see, the MFTR model not only inherits the bimodality and asymptotic decay properties exhibited separately by the FTR and $\kappa$-$\mu$ shadowed models, respectively, but also brings out additional flexibility to model propagation features not captured by the models from which it is originated. The resulting formulations of the MFTR statistics are as tractable as those of the simpler baseline models, and of other fading models in the state-of-the-art. The key contributions of this work can be listed as follows:

\begin{itemize}
    \item We derive the chief probability functions for the MFTR model, i.e., probability density function (PDF), cumulative {distribution} function (CDF), and moment {generating} function (MGF) in closed-form. These expressions allow for the computation of the MFTR statistics using special functions similar to those used in simpler, well-established fading models such as Rician shadowed \cite{Abdi2003}, $\kappa$-$\mu$ shadowed \cite{kms}, or FTR \cite{ftr}.
    \item Aiming to facilitate the use of the MFTR fading model for performance analysis purposes, we provide two alternative analytical formulations for the PDF and the CDF: one in terms of a continuous mixture of $\kappa$-$\mu$ shadowed distributions, and another one as an infinite {discrete} mixture of Gamma distributions. This allows to leverage the rich literature of performance analysis, so that existing results available either for the $\kappa$-$\mu$ shadowed or the Nakagami-$m$ cases can be used as a starting point to \textit{straightforwardly} analyze the performance under MFTR fading.
    \item {
    To experimentally validate the suitability of our model for experimental data, we carried out a fit of the MFTR model to small-scale fading field measurements in the terahertz (THz) band for outdoor line-of-sight (LOS) and non-line-of-sight (NLOS) environments.    }
    \item These results are used to analyze {$i)$ the outage probability, $ii)$ the average bit error rate (BER), and $iii)$ the ergodic capacity (EC), under MFTR fading channels. The former two metrics are derived in exact and asymptotic form. Also, we} study the impact of the model parameters on the amount of fading (AoF) metric.
\end{itemize}

%Secondly, we use the probability functions of the MFTR model to assess the performance of wireless communications systems under MFTR fading. Specifically, we are able to compute classical performance metrics such as the outage capacity probabilityΩ asymptotic outage capacity, and amount of fading. The validity of the expressions for such metrics is checked by Monte Carlo simulations by assuming a variety of fading conditions.

%Finally, we would like to highlight that, to the best of the author's knowledge, the MFTR model is the most general and unifying fading channel introduced in the literature with a reasonable mathematical tractability. We believe that this novel model will constitute a prolific research topic for the wireless communications field. }

%%%%%%%%%%%%%%%%%%%%%%%%%%%%%%%%
%Organization
The remainder of this paper is organized as follows: preliminaries and channel models are described in Section~\ref{sec2}. In Section~\ref{sec3}, analytical expressions are derived for the main statistics of the MFTR model. {In section \ref{secVal}, we introduce the empirical validation of our fading model
by fitting the MFTR model to a small-scale fading measurements campaign in outdoor THz wireless channels.} Then, in Section~\ref{sec5}, performance analysis over MFTR fading channels is exemplified. % we derive closed-form expressions of widely used metrics in wireless environments. 
Section~\ref{sec6} shows illustrative numerical results and discussions. Finally, concluding remarks are provided in Section~\ref{sec7}.

% Notation and terminology
\textit{Notation}: In what follows, $f_{(\cdot)}(\cdot)$ and $F_{(\cdot)}(\cdot)$ denote the PDF and CDF, respectively; $\mathbb{E}\left \{ \cdot  \right \}$ and $\mathbb{V}\left \{ \cdot  \right \}$ are the expectation and variance operators; $\Pr\left \{ \cdot  \right \}$ represents probability; $\abs{\cdot}$ is the absolute value, $\simeq$ refers to ``asymptotically equal~to'', $\equiv$ reads as ``is equivalent~to'', $\approx$ denotes ``approximately equal~to'' {and $\sim$ refers to ``statistically distributed as''}. 
In addition, $\Gamma(\cdot)$ denotes the gamma function~\cite[Eq.~(6.1.1)]{Abramowitz}, $\gamma(\cdot,\cdot)$ is the lower incomplete gamma function~\cite[Eq.~(6.5.2)]{Abramowitz}, $\left ( a \right )_ n =\tfrac{\Gamma\left ( a+n \right )}{\Gamma\left ( a \right )}$ represents the Pochhammer's symbol \cite{Abramowitz}, { $\bold{U}\left ( \cdot,\cdot,\cdot \right )$ is the confluent hypergeometric
Kummer function \cite[Eq.~(9.211)]{Gradshteyn}}, ${}_2F_1\left(\cdot,\cdot;\cdot;\cdot\right)$ is the Gauss hypergeometric function~\cite[Eq.~(15.1.1)]{Abramowitz}, $P_\alpha(z)={}_2F_1\left(-\alpha,\alpha+1;1;\tfrac{1-z}{2}\right)$ is the Legendre function of the first kind of real degree
$\alpha$~\cite[Eq.~(8.1.2)]{Abramowitz}, $F_D^{(4)}\left(\cdot,\cdot,\cdot,\cdot,\cdot;\cdot;\cdot,\cdot,\cdot,\cdot \right)$ denotes the Lauricella hypergeometric function in four variables~\cite[Eq.~(4), p. 33]{Srivastava}, $\Phi_2^{(4)}\left(\cdot,\cdot,\cdot,\cdot;\cdot;\cdot,\cdot,\cdot,\cdot \right)$ denotes the confluent hypergeometric function in four variables~\cite[Eq.~(8), p. 34]{Srivastava}, and $\Phi_2\left(\cdot,\cdot;\cdot;\cdot,\cdot \right)$ is the bivariate confluent hypergeometric function~\cite[Eq.~(4.19)]{Brychkov}.

%%% Preliminaries
\section{Preliminaries and Channel Models}\label{sec2}
According to~\cite{durgin}, the small-scale random fluctuations of a radio signal transmitted over a wireless channel can be structured at the receiver as a superposition of $N$ waves arising from dominant (specular) components, {plus a group of multipath} waves associated to diffuse scattering. Therefore, under this model, the received complex baseband signal representing the wireless channel can be expressed as
\begin{align}
\label{eq1}  
V_r =\sum_{n=1}^{N} V_n \exp\left({j\phi_n}\right)+X_1+jY_1 ,
\end{align}
where $V_n \exp(j\phi_n)$ denotes the $n$th specular component with constant amplitude $V_n$ and uniformly distributed random phase $\phi_n \sim \mathcal{U}(0, 2\pi)$. 
On the other hand, $X_1+jY_1$ is a {circularly-symmetric} complex Gaussian random variable (RV) with total power $2\sigma^2$, such that $X_1,Y_1 \sim \mathcal{N}(0,\sigma^2)$, representing the scattering components associated to the NLOS propagation. This model allows to individually account for a number of dominant waves, together with the application of the CLT to the diffuse component, where a sufficiently large number {of weak diffuse waves with independent phases is assumed.}

The FTR channel model was proposed in~\cite{ftr} as a generalization of the TWDP model \cite{durgin}, where the latter arises when considering two dominant specular components (i.e., $N=2$) in \eqref{eq1}. For its definition, the FTR model considers that the amplitudes of these two dominant components experience a joint fluctuation due to natural situations in different wireless scenarios (e.g., human body shadowing do to user motion, electromagnetic disturbances, and many others). Based on this, the complex signal under FTR fading can be formulated {as}~\cite{ftr}
\begin{align}
\label{eq2}  
V_{\rm{FTR}} =\sqrt{\zeta } V_1 \exp\left({j\phi_{{1}}}\right)+\sqrt{\zeta } V_2 \exp\left({j\phi_2}\right)+X_1+jY_1 ,
\end{align}
where $\zeta $ is a unit-mean Gamma distributed RV
whose PDF is given by
\begin{align}\label{eq3}
f_{\zeta}(x)=\frac{m^{m}x^{m-1}}{\Gamma(m)}\exp\left ( -m x \right ),
\end{align}
where $m$ denotes the shadowing severity index of the specular components.
Note that when $m\rightarrow\infty$, then $\zeta$ degenerates into a deterministic value and the amplitudes of the two dominant specular components in~\eqref{eq2} become constant. The FTR model in~\eqref{eq2}, besides fitting well with the field measurements in different wireless scenarios, also encompasses important statistical wireless channel models as particular cases. For instance, when no specular components are present in~\eqref{eq2}, i.e., $N=0$, the classical Rayleigh fading model arises. For a single LOS component, i.e., $N=1$, two fading models are obtained, namely, Rician and Rician shadowed~\cite{Abdi2003} for constant and fluctuating amplitudes, respectively. Finally, for the case in which there {are} two dominant components (i.e., $N=2$) with constant amplitudes,~\eqref{eq2} reduces to the TWDP fading model, also referred to as the Generalized Two-Ray fading model with uniformly distributed phases (GTR-U)~\cite{Rao2015}.

On the other hand, power-envelope based formulations as those originally proposed by Yacoub \cite{kappamuYacub} are defined from a different approach. Specifically, the squared amplitude (or instantaneous received power) of the $\kappa$-$\mu$ shadowed fading model is expressed as \cite{kms}
\begin{equation}
\label{eqclusters}
    R^2=\sum_{i=1}^{\mu}|Z_i+\sqrt{\zeta} p_i|^2,
\end{equation}
where $\mu$ wave clusters are defined, the complex variables $Z_i$ denote the diffuse components associated to each cluster, $\zeta$ is a unit-mean Gamma distributed RV, as given in \eqref{eq3}, and $p_i$ are complex amplitudes for the dominant components within each cluster. Notice that the FTR model in \eqref{eq2} can be physically interpreted as a single cluster in which both the specular and the diffuse components are part of the same cluster structure. With this in mind, we can combine a power-envelope definition as the one in \eqref{eqclusters} with the ray-based structure in \eqref{eq2} as follows.

%must also point out that the aforementioned cluster plus additional clusters\footnote{It is worth mentioning that, in the context of power-envelope or cluster-based models, a plethora of stochastic fading models have been proposed in the literature. Here, we are inspired by the $\kappa$-$\mu$ shadowed \cite{kms} model because it has proven to characterize several wireless scenarios accurately as well as including the famous Rician shadowed and $\kappa$-$\mu$ \cite{kappamuYacub} fading models.} composed by only NLOS waves can be a suitable model to accommodate to channel conditions encountered in some applications \comment{for} upcoming communications~\cite{clustersrappaport,Cheng2020,Tekbiyik2021}.
As in \cite{kappamuYacub}, we consider a wireless signal composed of clusters of waves propagating in a non-homogeneous environment. Within each cluster, the scattered waves have random phases and similar delay times, while the intercluster delay-time spreads are assumed to be relatively large. All clusters of the multipath waves are assumed to have scattered waves with identical powers. Now, in the first cluster (which typically represents the first one arriving), two dominant specular components with random phases and arbitrary power are considered, as in \eqref{eq2}, whereas in the rest of the clusters a specular component may also be present. Similarly to the models in \cite{ftr} and \cite{kms}, a dominant components are subject to the same source of random fluctuations. 
Under this channel model, the squared amplitude of the received signal is expressed as 
\begin{align}
\label{eq4}  
 R^2=\abs{ \underset{\rm{cluster \   1:}V_{\rm{FTR}} }{\underbrace{  \sqrt{\zeta }\left (V_1e^{j\phi_1}+ V_2e^{j\phi_2} \right )+Z_1}}}^2 + \underset{\rm{ \textcolor{black}{additional \ clusters } } }{\underbrace{ \sum_{i=2}^{\mu}\abs{\sqrt{\zeta } U_ie^{j\varphi_i}+Z_i}^2},}
\end{align}
 where $Z_i=X_i+jY_i$, for $i=\left \{1,\ldots,\mu  \right \}$ in which $X_i$ and $Y_i$ are mutually independent zero-mean Gaussian processes with $\sigma^2$ variance, i.e., $\mathbb{E}\{X_i^2\}=\mathbb{E}\{Y_i^2\}=\sigma^2$. 
%The insightful interpretation of~\eqref{eq4} is the following. Each additional multipath cluster is modelled by one term of the sum; thus, $\mu$ is the total number of multipath clusters. The scattered components (e.g., NLOS waves) of the $i$th cluster are denoted by a {circularly-symmetric} complex Gaussian RV, i.e., $Z_i$. Hence, for the $i$th cluster, the total power of the scattered multipath signals is $2\sigma^2$. 
In cluster 1, the complex RVs represented by  $V_n \exp(j\phi_n)$, for $n=\left \{1,2  \right \}$ denote the dominant specular components of the first arriving cluster, %; thus, the total power of them is given by $V_1^2+V_2^2$. 
whereas $U_i \exp(j\varphi_i)$ denotes the specular component of the $i$th cluster with constant amplitude $U_i$ and uniformly distributed random phase $\varphi_i \sim \mathcal{U}(0, 2\pi)$.
All the clusters are subject to the same source of random fluctuations as in the FTR and \mbox{$\kappa$-$\mu$} shadowed models, denoted by the normalized RV $\zeta$. For the model in~\eqref{eq4}, we coin the name multicluster FTR (MFTR) model, to indicate the presence of additional multipath clusters in the original FTR model. 

Because of the way it has been defined, the MFTR model defined in \eqref{eq4} includes \emph{both} the FTR and the $\kappa$-$\mu$ shadowed models as special cases, which have independently been validated, empirically  matching  different wireless scenarios \cite{ftr,kms}, which guarantees the practical usefulness of the proposed model.

As in \cite{kappamuYacub,kms}, even though the cluster number $\mu$ is inherently a natural number, the MFTR model admits a generalization for $\mu\in\mathbb{R}^+$, which is considered in the subsequent derivations.

%%%%%%%%%%%%%%%%%%%%%%%%%%%%%%%%
%FUNDAMENTAL STATISTICS
\section{Statistical Characterization of the {MFTR Fading Model}}\label{sec3}
In this section, the key first-order statistics of the newly proposed MFTR model are derived: the {MGF, PDF and CDF}. In the following derivations, let us consider {the} received power signal of the MFTR model, i.e., $W=R^2$, which {from~\eqref{eq4}} can be rewritten as  
\begin{align}
\label{eq5}  
 W&={{\left( \sqrt{\zeta }\left( {{V}_{1}}\cos {{\phi }_{1}}+{{V}_{2}}\cos {{\phi }_{2}} \right)+{{X}_{1}} \right)}^{2}} \nonumber \\ &+{{\left( \sqrt{\zeta }\left( {{V}_{1}}\sin {{\phi }_{1}}+{{V}_{2}}\sin {{\phi }_{2}} \right)+{{Y}_{1}} \right)}^{2}} 
  \nonumber\\&+ {
 \sum_{i=2}^{\mu} \left[\left( {\sqrt \zeta  U_i \cos \varphi _i  + X_i } \right)^2  + \left( {\sqrt \zeta  U_i sen\varphi _i  + Y_i } \right)^2\right].
 }
\end{align}
As in \cite{ftr}, the MFTR model can be conveniently expressed by introducing the parameters $K$ and $\Delta$, which are respectively defined as
\begin{align}
\label{eq6}
    K= \frac{V_{1}^{2}+V_{2}^{2}+\sum\limits_{i = 2}^\mu  U^2_i}{2{{\sigma }^{2}}\mu },
\end{align}
\begin{align}
\label{eq7}
    \Delta = \frac{2{{V}_{1}}{{V}_{2}}}{V_{1}^{2}+V_{2}^{2}+\sum\limits_{i = 2}^\mu  U^2_i}.
\end{align}
The MFTR fading model is univocally defined by four shape parameters: $\{K,m,\mu\}\in\mathbb{R}^+$ and $\Delta\in[0,1]$. Similar to the interpretation of the Rician $K$ factor, $K$ represents the ratio of the average power of the specular components to the power of the remaining scattered components. On the other hand, the $\Delta$ parameter ranging from 0 to 1 shows how similar to each other are the average received powers of the dominant components in cluster 1 and how much power is allocated to the specular components of the rest of the clusters. Thus, when only one specular component is present in cluster 1 ($V_1$ or $V_2 = 0$) then $\Delta = 0$, and the MFTR model collapses to the $\kappa$-$\mu$ shadowed model. When the magnitudes of the two specular {components} of the first cluster are equal ($V_1 =V_2$) and the remaining clusters lack specular components ($U_i=0$ for $i=\left \{2,\ldots,\mu  \right \}$) then $\Delta= 1$. Furthermore, notice that for $\mu=1$ {the MFTR model} yields the FTR fading model and the definitions of $K$ and $\Delta$ of both models coincide.

Next, we introduce the distribution of the received power signal (or equivalently the instantaneous signal-to-noise ratio (SNR) when the noise comes into play) of the MFTR fading model. It is worth mentioning that the statistical characterization of the instantaneous received SNR, here denoted by $\gamma$, plays a pivotal role in designing or evaluating the performance of many practical wireless systems. {Let $\gamma \triangleq (E_s / N_0) W$ be the instantaneous received SNR through  the MFTR fading channel, with $E_s$ and $ N_0$ representing, respectively, the energy density per symbol and the power spectral density. 
\textcolor{black}{
Mathematically speaking, {from \eqref{eq5}}, $\overline{\gamma}$ can be formulated as
\begin{align}
\label{eq8}
   \overline{\gamma} =&\frac{E_s}{N_0} \mathbb{E}\left \{ W\right \}
	\nonumber \\ &{=\frac{E_s}{N_0}\Biggr( \mathbb{E}\left \{ \abs{ {  \sqrt{\zeta }\left (V_1e^{j\phi_1}+ V_2e^{j\phi_2} \right )+X_1+jY_1}}^2 \right \} }
   \nonumber \\ & {+ \sum_{i=2}^{\mu}\mathbb{E}\left \{\abs{\sqrt{\zeta }U_ie^{j\varphi_i}+X_i+jY_i}^2 \right \}\Biggr)} 
   \nonumber \\ =& {\frac{E_s}{N_0} \left ( V_1^2+V_2^2+2\sigma^2+\sum_{i=2}^{\mu}\left(U^2_i+2\sigma^2\right)\right )} 
   \nonumber \\ =& \frac{E_s}{N_0}\left(V_1^2+V_2^2+\sum_{i=2}^{\mu}U^2_i + 2\sigma^2\mu\right)
   \nonumber \\ { \stackrel{(a)}{=}}&   \frac{E_s}{N_0} 2\sigma^2\mu(1+K).
\end{align}}
{where in step (a), we employ \eqref{eq6} with the respective substitutions.}
With the aid of the previous definitions, the chief probability functions concerning the MFTR channel model can be derived as follows.

\subsection{MGF}

{In the first Lemma, presented below, we obtain a closed-form expression for the MGF.}

\begin{lemma}\label{lemma1}
The {MGF} of the instantaneous received SNR $\gamma$ {under MFTR fading} can be expressed as
\begin{align}
\label{eq9}
   {{{\mathcal{M}}}_{\gamma }}\left( s \right)=&\frac{{{m}^{m}}{{\mu }^{\mu }}{{\left( 1+K \right)}^{\mu }}{{\left( \mu \left( 1+K \right)-\overline{\gamma }s \right)}^{m-\mu }}}{{{\left( \sqrt{\mathcal{R}(\mu,m,K,\Delta;s)} \right)}^{m}}} \nonumber \\ & \times 
   {{P}_{m-1}}\left( \frac{m\mu \left( 1+K \right)-\left( \mu K+m \right)\overline{\gamma }s}{\sqrt{\mathcal{R}(\mu,m,K,\Delta;s)}} \right)
\end{align}
\end{lemma}
where $\mathcal{R}(\mu,m,K,\Delta;s)$ is a polynomial in $s$ given by 
\begin{align}
\label{eq10}
  \mathcal{R}(\mu,m,K,\Delta;s)&=\left[ {{\left( m+\mu K \right)}^{2}}-{{\left( \mu K\Delta  \right)}^{2}} \right]{{\overline{\gamma }}^{2}}{{s}^{2}} -2m\mu \nonumber \\ & \times
 \left( 1+K \right) \left( m+\mu K \right)\overline{\gamma }s+{{\left[ \left( 1+K \right)m\mu  \right]}^{2}}.
\end{align}
\begin{proof}
	See Appendix~\ref{ap:MGF}.
\end{proof}
{Note that the result in Lemma \ref{lemma1} is valid for any positive real value of $m$.}

\subsection{PDF and CDF}
Here, we derive the PDF and CDF of the SNR of the MFTR model\footnote{The PDF and CDF of the received signal envelope $R$ under MFTR channels can be obtained through a standard variable transformation, yielding $f_R(r)=2r f_\gamma(r^2)$ and $F_R(r)=F_\gamma(r^2)$, with $\overline{ \gamma}$ replaced by $\Omega=\mathbb{E}\left \{ R^2  \right \}$.}. Even though these can be computed by performing a numerical Laplace inverse transform over the MGF in Lemma \ref{lemma1} for any arbitrary set of values of the shape parameters \cite{Simon}, it is possible to obtain closed-form expressions by assuming that the fading parameter $m$ takes integer values. For this purpose, we take advantage of the fact that, for $m\in \mathbb{Z}^+$, the Legendre function in the MGF obtained in (\ref{eq9}) has an integer degree;
{thus, such Legendre function becomes a Legendre polynomial}. The Legendre polynomial of an integer degree $n$ can be formulated as in~\cite[Eq.~(22.3.8)]{Abramowitz} by
\begin{align}
\label{eq11}
P_n(z)=\frac{1}{2^n}\sum_{q=0}^{\left \lfloor n/2 \right \rfloor}(-1)^nC_q^nz^{n-2q},
\end{align}
where the $C_q^n$ term is expressed as
\begin{align}
\label{eq12}
C_q^n=\binom{n}{q}\binom{2n-2q}{n}=\frac{\left ( 2n-2q \right )!}{q!(n-q)!(n-2q)!}.
\end{align}
{From the MGF in~\eqref{eq9} and with the help of~\eqref{eq11}, the closed-form expressions for the PDF and CDF of the RV $\gamma$ are obtained in the following Lemma.}
\begin{lemma}\label{lemma2}
Assuming that $m\in \mathbb{Z}^+$, the PDF and CDF of the SNR under MFTR fading can be formulated as \eqref{eq13} and \eqref{eq14}, respectively. 
\end{lemma}
%\begin{align}
%\label{eq13}
 % {{f}_{\gamma }}\left( x \right)=&\frac{2\left( a_2 a_3 \right)^{\frac{m}{2}}x^{\mu-1}}{\Gamma\left( \mu \right)2^{m} a_4^{m-\mu}}\sum_{q=0}^{  \left \lfloor \frac{m-1}{2} \right \rfloor}\frac{C_q^{m-1}  }{(-1)^{-q}}\left[ \frac{a_2a_3}{a_1^2} \right]^{\frac{c_1}{2}} \Phi_2^{(4)}(-c_1,
 %   \nonumber \\ & c_2,c_2,\mu-m;\mu;-a_1x,-a_2x,-a_3x,-a_4x ),
%\end{align}
%\begin{align}
%\label{eq14}
 % {{F}_{\gamma }}\left( x \right)=&\frac{2a_4^\mu\left( a_2 a_3 \right)^{\frac{m}{2}}x^{\mu}}{\Gamma\left( \mu +1\right)2^{m} a_4^{m}}\sum_{q=0}^{  \left \lfloor \frac{m-1}{2} \right \rfloor}\frac{C_q^{m-1}  }{(-1)^{-q}}\left[ \frac{a_2a_3}{a_1^2} \right]^{\frac{c_1}{2}}\Phi_2^{(4)}(-c_1,
    %\nonumber \\ & c_2,c_2,\mu-m;\mu+1;-a_1x,-a_2x,-a_3x,-a_4x ).
%\end{align}
%where
%\begin{align}
%\label{eq15}
% a_1=&\frac{\left( 1+K \right)m\mu }{\left( m+\mu K \right)\overline{\gamma }}, \quad a_2=\frac{\left( 1+K \right)m\mu }{\left( m+\mu K\left( 1+\Delta  \right) \right)\overline{\gamma }}, \nonumber \\ 
% a_3=&\frac{\left( 1+K \right)m\mu }{\left( m+\mu K\left( 1-\Delta  \right) \right)\overline{\gamma }}, \quad   a_4=\frac{\left( 1+K \right)\mu }{\overline{\gamma }}, \nonumber \\ 
% c_1=&m-1-2q,\quad  c_2=m-q-1/2.
%\end{align}
\begin{proof}
	See Appendix~\ref{ap:PFDCDF}.
\end{proof}
The chief statistics of the MFTR fading model in~\eqref{eq13} and~\eqref{eq14} are given in terms of the {multi-variate confluent hypergeometric function} $\Phi_2^{(4)} \left(\cdot \right)$, which is rather common in well-established fading models such as Rician {shadowed} \cite{Abdi2003}, $\kappa$-$\mu$ shadowed \cite{kms}, or FTR \cite{ftr}. Moreover, the computation of this function can be performed by resorting to an inverse Laplace transform as described in~\cite[Appendix~9B]{Simon}, and whose implementation in a simple and efficient way through MATLAB is given in~\cite{phi2function}. Therefore, the evaluation of MFTR probability distributions does not pose any additional challenge compared to other well-known fading models in the state-of-the-art.

\subsection{Alternative formulations}
Expressions in Lemmas \ref{lemma1} and \ref{lemma2} provide a complete formulation for the MFTR, equivalent in complexity to those originally proposed in \cite{ftr} and \cite{kms} for the baseline FTR and $\kappa$-$\mu$ shadowed fading distributions, respectively. However, aiming to provide additional flexibility to the newly proposed MFTR model, as well as to facilitate its use for performance evaluation purposes, we now provide two alternative formulations for the PDF and CDF of the MFTR model.

%Specifically, based on the fact that the MFTR distribution can be expressed in terms of an underlying $\kappa$-$\mu$ shadowed distribution when conditioned on $\theta$ (defined in Appendix~\ref{ap:MGF}), 
We first propose a formulation of the MFTR model as a \textit{continuous} mixture of $\kappa$-$\mu$ shadowed distributions. Secondly, %based on the fact that the MFTR distribution can be expressed in terms of an underlying FTR distribution when conditioned on $\zeta$, 
we propose a formulation of the MFTR model as an infinite \textit{discrete} mixture of Gamma distributions. These formulations are provided in the following two lemmas, and are valid for the entire range of values of the shape parameters $\kappa, \mu, m$ and $\Delta$. 

%%%%%%%%%%%%%%%%%%%%%%%%%%%%%%%%%%%%%%%%%%%%%%%%%%%%%%%  PDF for integer m
\begin{figure*}[ht]
	%\hrulefill
	\begin{normalsize}
\begin{align}\label{eq13}
 f_{\gamma }( x)=&\frac{\left ( 1+K \right )^{\mu}\mu ^{\mu}}{2^{m-1}\Gamma(\mu)\overline{\gamma}^\mu}\left ( \frac{m}{\sqrt{\left ( m+\mu K \right )^2-\mu^2 K^2 \Delta^2}} \right )^m\sum_{q=0}^{  \left \lfloor \frac{m-1}{2} \right \rfloor}\left( -1 \right)^q C_q^{m-1} \left ( \frac{m+\mu K}{\sqrt{\left ( m+\mu K \right )^2-\mu^2 K^2 \Delta^2}} \right )^{m-1-2q}  \nonumber \\  \times & \ x^{\mu-1} \Phi_2^{(4)}\Biggl(1+2q-m,m-q-1/2,m-q-1/2,\mu-m;\mu;-\frac{\left( 1+K \right)m\mu }{\left( m+\mu K \right)\overline{\gamma }}x,-\frac{\left( 1+K \right)m\mu }{\left( m+\mu K\left( 1+\Delta  \right) \right)\overline{\gamma }}x, \nonumber \\ &
 -\frac{\left( 1+K \right)m\mu }{\left( m+\mu K\left( 1-\Delta  \right) \right)\overline{\gamma }}x,-\frac{\left( 1+K \right)\mu }{\overline{\gamma }}x \Biggl).
\end{align}
	\end{normalsize}
%	\hrulefill
	\vspace{-5mm}
\end{figure*}
%%%%%%%%%%%%%%%%%%%%%%%%%%%%%%%%%%%%%%%%%%%%%%%%%% 
%%%%%%%%%%%%%%%%%%%%%%%%%%%%%%%%%%%%%%%%%%%%%%%%%%%%%%%  CDF for integer m
\begin{figure*}[ht]
	%\hrulefill
	\begin{normalsize}
\begin{align}\label{eq14}
 F_{\gamma }( x)=&\frac{ \left ( 1+K \right )^{\mu}\mu ^{\mu}}{2^{m-1}\Gamma(\mu+1) \overline{\gamma}^\mu}\left ( \frac{m}{\sqrt{\left ( m+\mu K \right )^2-\mu^2 K^2 \Delta^2}} \right )^m\sum_{q=0}^{  \left \lfloor \frac{m-1}{2} \right \rfloor}\left( -1 \right)^q C_q^{m-1} \left ( \frac{m+\mu K}{\sqrt{\left ( m+\mu K \right )^2-\mu^2 K^2 \Delta^2}} \right )^{m-1-2q}  \nonumber \\ \times &\ x^{\mu} \Phi_2^{(4)}\Biggl(1+2q-m,m-q-1/2,m-q-1/2,\mu-m;\mu+1;-\frac{\left( 1+K \right)m\mu }{\left( m+\mu K \right)\overline{\gamma }}x,-\frac{\left( 1+K \right)m\mu }{\left( m+\mu K\left( 1+\Delta  \right) \right)\overline{\gamma }}x, \nonumber \\ &
 -\frac{\left( 1+K \right)m\mu }{\left( m+\mu K\left( 1-\Delta  \right) \right)\overline{\gamma }}x,-\frac{\left( 1+K \right)\mu }{\overline{\gamma }}x \Biggl).
\end{align}
	\end{normalsize} 
	\hrulefill
	\vspace{-5mm}
\end{figure*}
%%%%%%%%%%%%%%%%%%%%%%%%%%%%%%%%%%%%%%%%%%%%%%%%%%

%%%%%%%%%%%%%%%%%%%%%%%%%%%%%%%%%%%%%%%%%%%%%%%%%%  LEMMA 
\begin{lemma}\label{lemma4}
When $m\in \mathbb{R}^+$, the PDF and CDF of the SNR of the MFTR distribution can be obtained by averaging the conditional $\kappa$-$\mu$ shadowed statistics over all possible
realizations of $\theta$, as% in the same way as indicated in \eqref{Lema1eq14}, i.e., 
\end{lemma}
\begin{align}\label{eq18}
f_\gamma(x)=\frac{1}{\pi}\int_{0}^{\pi}{{f}_{\gamma \left| \theta  \right.}}(x)d\theta,
 \end{align}
\begin{align}\label{eq19}
F_\gamma(x)=\frac{1}{\pi}\int_{0}^{\pi}{{F}_{\gamma \left| \theta  \right.}}(x)d\theta,
 \end{align}
where
\begin{align}
\label{eq20}
    	&{{f}_{\gamma \left| \theta  \right.}}(x)=\frac{{{\mu }^{\mu }}{{m}^{m}}{{\left( 1+K \right)}^{\mu }}}{\Gamma (\mu )\overline{\gamma }{{\left( \mu K\left( 1+\Delta \cos \theta  \right)+m \right)}^{m}}}{{\left( \frac{x}{\overline{\gamma }} \right)}^{\mu -1}} \nonumber \\ & \times {{e}^{-\frac{\mu \left( 1+K \right)}{\overline{\gamma }}x}} 
    	{{}_{1}}{{F}_{1}}\left( m;\mu ;\frac{{{\mu }^{2}}K\left( 1+\Delta \cos \theta  \right)\left( 1+K \right)}{\mu K\left( 1+\Delta \cos \theta  \right)+m}\frac{x}{\overline{\gamma }} \right){,}
\end{align}
\begin{align}
\label{eq21}
    	{{F}_{\gamma \left| \theta  \right.}}(x)=&\frac{{{\mu }^{\mu-1 }}{{m}^{m}}{{\left( 1+K \right)}^{\mu }} }{\Gamma (\mu ){{\left( \mu K\left( 1+\Delta \cos \theta  \right)+m \right)}^{m}}}{{\left( \frac{x}{\overline{\gamma }} \right)}^{\mu}} \Phi_2\Bigl(\mu-m, \nonumber \\ & m;\mu+1;-\frac{\mu\left ( 1+K \right )x}{\overline{\gamma}},-\frac{\mu\left ( 1+K \right )}{\overline{\gamma}}
    	\nonumber \\ & \times 
    	\frac{m x}{\mu K\left (1+\Delta \cos\theta  \right )+m } \Bigl).
\end{align}
\begin{proof}
The conditional $\kappa$-$\mu$ shadowed PDF and CDF are obtained from~\cite[Eq.~(4)]{kms} and~\cite[Eq.~(6)]{kms} by substituting $\kappa$ and $\overline{\gamma}$ by~\eqref{Lema1eq4} and~\eqref{Lema1eq5}, respectively. Then, using the relationships given in~\eqref{Lema1eq7} and~\eqref{Lema1eq10} that connect the $\kappa$-$\mu$ shadowed and MFTR models,~\eqref{eq20} and~\eqref{eq21} are obtained.%attained immediately.
\end{proof}

It must be noted that the Rician shadowed distribution is a particular case of the $\kappa$-$\mu$ shadowed distribution for the case when $\mu=1$. Thus, the integral connection between the FTR and the Rician shadowed distributions presented in \cite{ftr2}, for arbitrary positive real $m$, is a particular case of Lemma~\ref{lemma4} for $\mu=1$.

%%%%%%%%%%%%%%%%%%%%%%%%%%%%%%%%%%%%%%%%%%%%%%%%%%%%%%%  PDF for real m
\begin{figure*}[h!]
	%\hrulefill
	\begin{normalsize}
\begin{align}\label{eq15}
 f_{\gamma }( x)=  \sum_{i=0}^{\infty}\textcolor{black}{w_i} f_X^{\rm G} \left(   \mu+i;\frac{\overline{\gamma}(\mu+i)}{\mu(K+1)} ;x\right).
\end{align}
	\end{normalsize}
%	\hrulefill
	\vspace{-5mm}
\end{figure*}
%%%%%%%%%%%%%%%%%%%%%%%%%%%%%%%%%%%%%%%%%%%%%%%
%%%%%%%%%%%%%%%%%%%%%%%%%%%%%%%%%%%%%%%%%%%%%%%%%%%%%%%  CDF for real m
\begin{figure*}[h!]
	%\hrulefill
	\begin{normalsize}
\begin{align}\label{eq16}
 F_{\gamma }( x)=  \sum_{i=0}^{\infty}\textcolor{black}{w_i} F_X^{\rm G}\left(   \mu+i;\frac{\overline{\gamma}(\mu+i)}{\mu(K+1)} ;x\right).
\end{align}
	\end{normalsize}
%	\hrulefill
	\vspace{-5mm}
\end{figure*}
%%%%%%%%%%%%%%%%%%%  PESOS DE MIXTURA GAMMA

\begin{figure*}[h!]
	%\hrulefill
	\begin{normalsize}
\begin{align}\label{eqWeights}
\textcolor{black}{
w_i=\frac{\Gamma(m+i)(\mu K)^i m^m}{\Gamma(m)\Gamma(i+1)}\frac{{\left( {1 - \Delta } \right)^i }}{{\sqrt \pi  (\mu K(1 - \Delta ) + m)^{m + i} }}}&\textcolor{black}{\sum_{q=0}^{i}\binom{i}{q}\frac{{\Gamma \left( {q + \frac{1}{2}} \right)}}{{\Gamma \left( {q + 1} \right)}} \left( {\frac{{2\Delta }}{{1 - \Delta }}} \right)^q} \nonumber \\  & \textcolor{black}{\times
{}_2F_1\left( {m + i,q + \frac{1}{2};q + 1;\frac{{ - 2\mu K\Delta }}{{\mu K(1 - \Delta ) + m}}} \right),\quad m \in R^+,}
\end{align}
	\end{normalsize}
%	\hrulefill
	\vspace{-5mm}
\end{figure*}

\begin{figure*}[h!]
	%\hrulefill
	\begin{normalsize}
\begin{align}\label{eq17}
 f_X^{\rm G}\left(  \lambda;\nu ;y\right)=\frac{ \lambda^ \lambda}{\Gamma( \lambda) \nu^ \lambda}y^{\lambda-1}\exp\left( - \frac{\lambda y}{\nu} \right), \quad   F_X^{\rm G}\left(  \lambda;\nu ;y\right)=\frac{1}{\Gamma(\lambda)}\gamma\left ( \lambda,\frac{\lambda y}{\nu} \right ).
\end{align}
	\end{normalsize}
	\hrulefill
	\vspace{-5mm}
\end{figure*}

\begin{lemma}\label{lemma3}
{When $m\in \mathbb{R}^+$, the PDF and CDF of the SNR of the MFTR distribution can be obtained as an infinite discrete mixture of Gamma distributions. The corresponding expressions are given in equations \eqref{eq15} and \eqref{eq16}, respectively, where $f_X^{\rm G}(\cdot)$ and $F_X^{\rm G}(\cdot)$ represent the PDF and CDF, respectively, of a Gamma distribution, which are both given in \eqref{eq17}. 
%Note that $\gamma(\cdot,\cdot)$ in \eqref{eq17} represents the incomplete gamma function defined in~\cite[Eq.~(6.5.2)]{Abramowitz}. %\cite{Gradshteyn} and the integral $I_2$
%that appears in \eqref{eq15} and \eqref{eq16} is derived in Appendix \ref{ap:asc}
}
\end{lemma}
\begin{proof}
	See Appendix~\ref{ap:PFDCDFreal}.
\end{proof}

Here, we point out that the PDFs and CDFs in Lemmas~\ref{lemma4} and~\ref{lemma3} are valid for non-constrained fading values of the MFTR model. Expressions in Lemma~\ref{lemma4}, i.e.,~\eqref{eq18} and~\eqref{eq19}, are given in simple finite-integral form in terms of well-known functions in communication theory, where the integrands are continuous bounded functions and the integration interval is finite. Therefore, the evaluation of these integrals through numerical integration routines in commercial mathematical software packages poses no challenge, and in fact is an standard approach in communication theory -- cfr. the proper integral forms of the Gaussian $Q$-function \cite{Weinstein1974}, or Simon and Alouini's MGF approach to the performance analysis of wireless communication systems \cite{Simon}. Expressions in Lemma~\ref{lemma3}, i.e.,~\eqref{eq15} and~\eqref{eq16}, are given as weighted sums of gamma distributions, which are a basic building block in many communication theory applications, and correspond to the case of assuming Nakagami-$m$ fading. These sets of expressions allow for leveraging the rich literature devoted to study baseline fading models like $\kappa$-$\mu$ shadowed and Nakagami-$m$, to directly evaluate the case of the more general MFTR model, when desired.

%%%%%%%%%%%%%%%%%%%%%%%%%%%%%%%%
\begin{table}[t]    
%\scriptsize
\footnotesize
  \caption{Conventional and generalized fading channel models derived from the MFTR distribution. }
\centering
\begin{tabular}{c|c}
	\hline \hline
\multicolumn{1}{l|}{	\textbf{Fading Distribution}} & \multicolumn{1}{c}{\textbf{MFTR fading parameters } } 
\\  \hline  \hline  
 & $a)$ $\underline{\Delta}=0$, $\underline{K} \rightarrow \infty$, $\underline{m}=0.5$, $\underline{\mu}=1$ 
\\ One-sided Gaussian       & {$b)$ $\underline{\Delta}=0$, $\underline{K} \rightarrow 0$, $\underline{m}\rightarrow\infty$, $\underline{\mu}=0.5$}
\\        & $c)$ $\underline{\Delta}=1$, $\underline{K} \rightarrow \infty$, $\underline{m}=1$, $\underline{\mu}=1$
\\ \hline   & $a)$$\underline{\Delta}=0$, $\underline{K} \rightarrow \infty$, $\underline{m}=1$, $\underline{\mu}=1$
\\ Rayleigh        &{ $b)$ $\underline{\Delta}=0$, $\underline{K}\rightarrow0$, $\underline{m}\rightarrow\infty$, $\underline{\mu}=1$}
\\         & $c)$ $\underline{\Delta}=0$, $\underline{K} =0$, $\forall \underline{m}$, $\underline{\mu}=1$
  \\ \hline  & $a)$ $\underline{\Delta}=0$, $\underline{K} =\tfrac{1-q^2}{2q^2}$, $\underline{m}=0.5$, $\underline{\mu}=1$
 \\ Nakagami-$q$ (Hoyt)  & $b)$ $\forall \left \{\underline{\Delta},\underline{K}  \right \}$, with $q=\sqrt{\tfrac{1+K(1-\Delta)}{1+K(1+\Delta)}}$, \\   & $\underline{m}=1$, $\underline{\mu}=1$
 \\ \hline Nakagami-$m$ & $a)$  $\underline{\Delta}=0$, $\underline{K} \rightarrow \infty$, $\underline{m}=m$, $\underline{\mu}=1$  \\  & $b)$ {$\underline{\Delta}=0$, $\underline{K} \rightarrow 0$, $\underline{m}\rightarrow \infty$, $\underline{\mu}=m$}
  \\ \hline Rician & $\underline{\Delta}=0$, $\underline{K}=K$, $\underline{m}\rightarrow \infty$, $\underline{\mu}=1$
    \\ \hline Rician shadowed & $\underline{\Delta}=0$, $\underline{K}=K$, $\underline{m}=m$, $\underline{\mu}=1$
     \\ \hline $\kappa$-$\mu$ shadowed & $\underline{\Delta}=0$, $\underline{K}=\kappa$, $\underline{m}=m$, $\underline{\mu}=\mu$
  \\ \hline $\kappa$-$\mu$ & $\underline{\Delta}=0$, $\underline{K}=\kappa$, $\underline{m}=\rightarrow \infty$, $\underline{\mu}=\mu$
  \\ \hline $\eta$-$\mu$
  & $\underline{\Delta}=0$, $\underline{K}=(1-\eta)/(2\eta)$,  $\underline{m}= \mu$, $\underline{\mu}=2 \mu$ 
  \\ \hline TWDP
  & $\underline{\Delta}=\Delta$, $\underline{K}=K$, $\underline{m}=\rightarrow \infty$, $\underline{\mu}=1$  \\ \hline Two-Wave
  & $\underline{\Delta}=\Delta$, $\underline{K}=\rightarrow \infty$, $\underline{m}=\rightarrow \infty$, $\underline{\mu}=1$  \\ \hline FTR
  & $\underline{\Delta}=\Delta$, $\underline{K}=K$, $\underline{m}=m$, $\underline{\mu}=1$  \\ \hline
\multicolumn{1}{c|}{Fluctuating Two-Wave} & \multicolumn{1}{c}{$\underline{\Delta}=\Delta$, $\underline{K}=\rightarrow \infty$, $\underline{m} = m$, $\underline{\mu}=1$} \\ \hline  \hline
\end{tabular}\label{Cases}
%\vspace{-8mm}
\end{table} 

% figures PDFs
\begin{figure}[t]
\centering 
\psfrag{A}[Bc][Bc][0.7]{\rm Case $\mathrm{A}$}
\psfrag{B}[Bc][Bc][0.7]{\rm Case $\mathrm{B}$}
\includegraphics[width=1\linewidth]{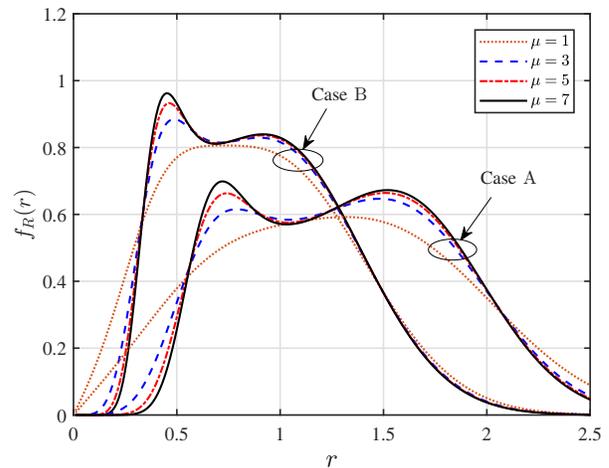} \caption{PDF of the MFTR signal envelope by varying $\mu$ in two different scenarios:{ \text{ Case $\mathrm{A}$}: $\Delta=0.9$, $m=8$, $K=8$, $ \overline{\gamma}=2$ and \text{Case $\mathrm{B}$}: $\Delta=0.9$, $m=4$, $K=15$,  $\overline{\gamma}=1$.} }
\label{fig1}
\vspace{-2mm}
\end{figure}

\begin{figure}[H]
\centering 
\psfrag{H}[Bc][Bc][0.6]{$K_{\mathrm{dB}}^{\mathrm{B}}=25$ dB}
\psfrag{Z}[Bc][Bc][0.6][0]{$K_{\mathrm{dB}}^{\mathrm{B}}=15$ dB}
\includegraphics[width=1\linewidth]{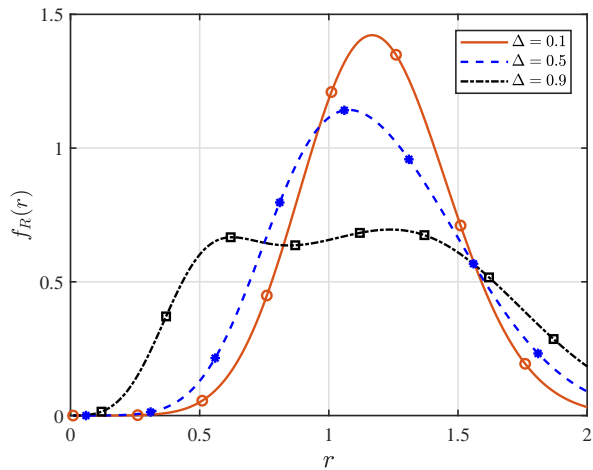} \caption{PDF of the MFTR signal envelope for different values of $\Delta$, with $\mu=2$, $m=6$, $K=15$, and $\overline{\gamma}=1.5$.  \textcolor{black}{Markers correspond to MC simulations.}}
\label{fig2}
\end{figure}

\begin{figure}[t]
\centering 
\psfrag{H}[Bc][Bc][0.6]{$K_{\mathrm{dB}}^{\mathrm{B}}=25$ dB}
\psfrag{Z}[Bc][Bc][0.6][0]{$K_{\mathrm{dB}}^{\mathrm{B}}=15$ dB}
\includegraphics[width=1\linewidth]{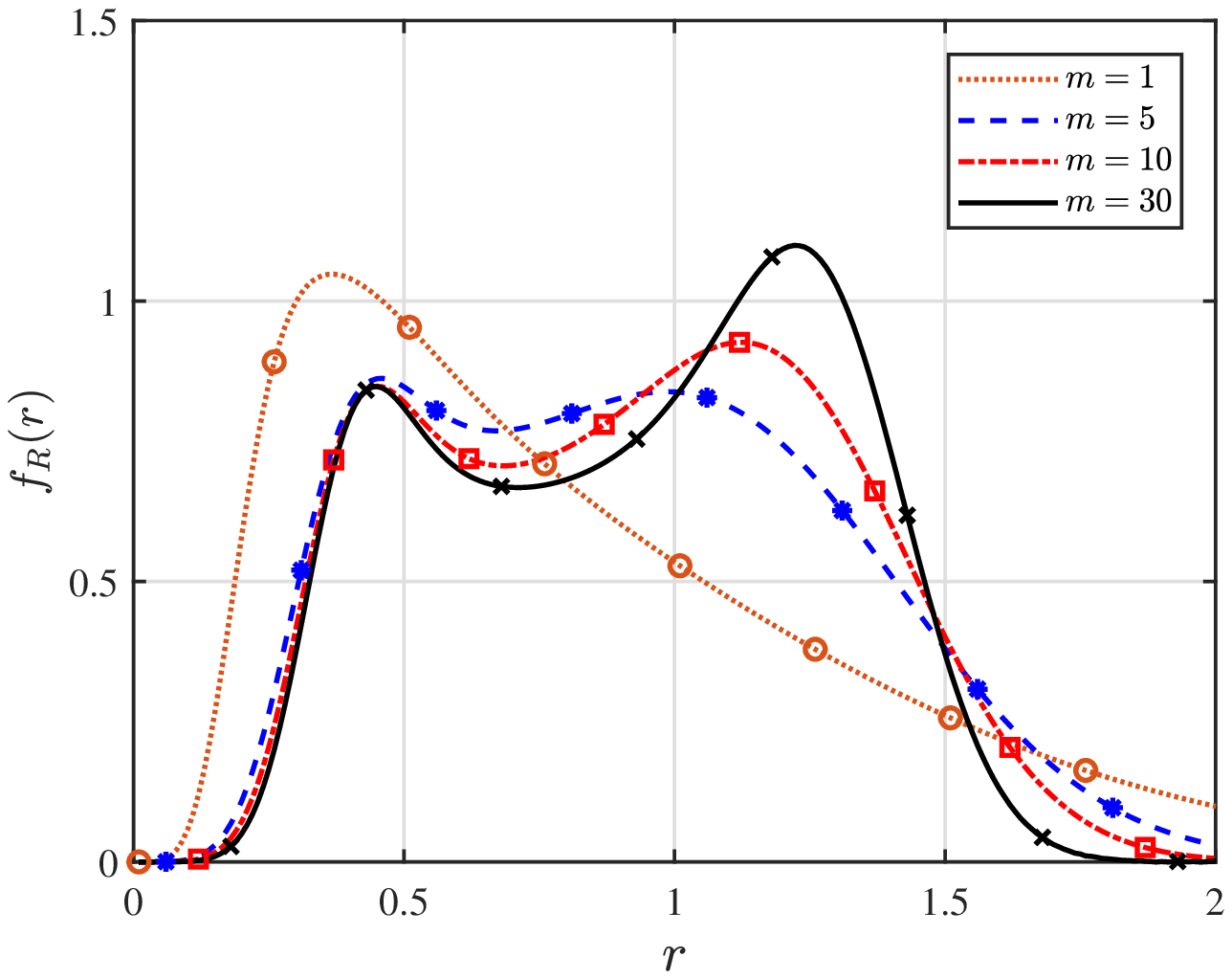} \caption{PDF of the MFTR signal envelope for different values of $m$, with $\Delta=0.9$, $\mu=3$, $m=4$, $K=20$, and $\overline{\gamma}=1$. \textcolor{black}{Markers correspond to MC simulations.}}
\label{fig3}
\end{figure}

\begin{figure}[t]
\centering 
\psfrag{H}[Bc][Bc][0.6]{$K_{\mathrm{dB}}^{\mathrm{B}}=25$ dB}
\psfrag{Z}[Bc][Bc][0.6][0]{$K_{\mathrm{dB}}^{\mathrm{B}}=15$ dB}
\includegraphics[width=1\linewidth]{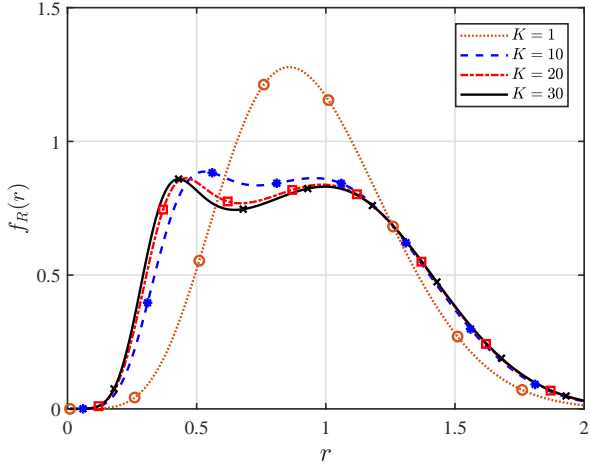} \caption{PDF of the MFTR signal envelope for different values of $K$, with $m=5$, $\Delta=0.9$, $\mu=3$, and $\overline{\gamma}=1$. \textcolor{black}{Markers correspond to MC simulations}}
\label{fig4}
\end{figure}
\subsection{Special cases and effect of parameters}
The MFTR model derived here is connected to other fading distributions commonly used in several wireless application scenarios, by specializing the corresponding set of parameters as stated in Table~\ref{Cases}. In order to avoid confusion, the parameters related to the MFTR distribution are underlined. We would like to mention that a multicluster version of the TWDP model in \cite{durgin,Rao2015} naturally appears as we let $m\rightarrow\infty$; however, this model alone has its own entity an deserves special attention as a separate item \cite{pena2022}.

In Figs.~\ref{fig1}-\ref{fig4}, we illustrate how different propagation conditions affect the shape of the MFTR distribution, by evaluating the PDF of the received signal envelope $f_R(r)$ for a set of values of the shape parameters: $K$, $\Delta$, $m$ and $\mu$. 
%$K$ (power ratio between dominant and diffuse components), $\Delta$ (amplitude imbalance between the two dominant components), $m$ (severity of dominant waves' fluctuation), and $\mu$ (clustering of scattered multipath waves). 
The PDFs illustrated in the figures have been obtained by evaluating \eqref{eq13}, although all mathematical expressions along this section have been double-checked through Monte Carlo (MC) simulations, which have been included with markers in the figures, whenever they don't affect readability. Also, we also checked that the same results are obtained when eqs. \eqref{eq18} and \eqref{eq15} are used to evaluate the PDFs.

%For the sake of legibility, here, \textcolor{red}{we chose to show only the analytical PDFs obtained from \eqref{eq13} and \eqref{eq15} because both the Monte Carlo simulations and the PDFs derived from \eqref{eq20} also computed to validate our results are practically all indistinguishable from each other.}

In Fig.~\ref{fig1}, we clearly perceive that the MFTR model's behavior is inherently bimodal\footnote{The bimodality of the distribution is related to the existence of two local maxima in its PDF.} {(see \textit{Case $\mathrm{A}$})}. Interestingly, thanks to the presence of the $\mu$ parameter adequately combined with the other fading parameters, the MFTR model can exhibit a more pronounced bimodality as $\mu$ increases {(see \textit{Case $\mathrm{B}$})}. Specifically, the MFTR model can exhibit both a left-bimodality (i.e., the first local maximum is larger) and a right-bimodality (i.e., the second local maximum is larger). This is in stark contrast with the behavior of the baseline $\kappa$-$\mu$ shadowed (unimodal) or FTR distributions (only right-bimodality) from which the MFTR distribution originates. This feature brings additional flexibility to improve the versatility of the MFTR model to fit field measurements in emerging wireless scenarios such as mm-Wave and sub-terahertz bands \cite{Du2022}. 

In Fig.~\ref{fig2}, it is confirmed that for low values of $\Delta$, i.e., one specular component is dominant in cluster 1, the MFTR distribution exhibits a unimodal behavior similar to the $\kappa$-$\mu$ shadowed case. Conversely, for larger values of $\Delta$, i.e., two specular components with similar amplitudes are present in cluster 1 which dominate the power of the specular components of the remaining clusters, the MFTR distribution exhibits a bimodal behavior. From Figs.~\ref{fig3}-\ref{fig4}, we can see that such bimodality of the MFTR distribution is also closely linked to parameters $m$ and $K$. Specifically, large values of $m$ or $K$ yield a more pronounced bimodality. Conversely, low values of $m$ or $K$ tend to smoothen such bimodality.

%%%%%%%%%%%%%%%%%%%%%%%%%%%%%%%%
%EMPIRICAL VALIDATION
\begin{figure}[t]
%\hspace{-1cm}
\centering
\subfigure[\hspace{-4cm}]{}\hspace{-0.4in}
{\includegraphics[width=0.28\textwidth]
{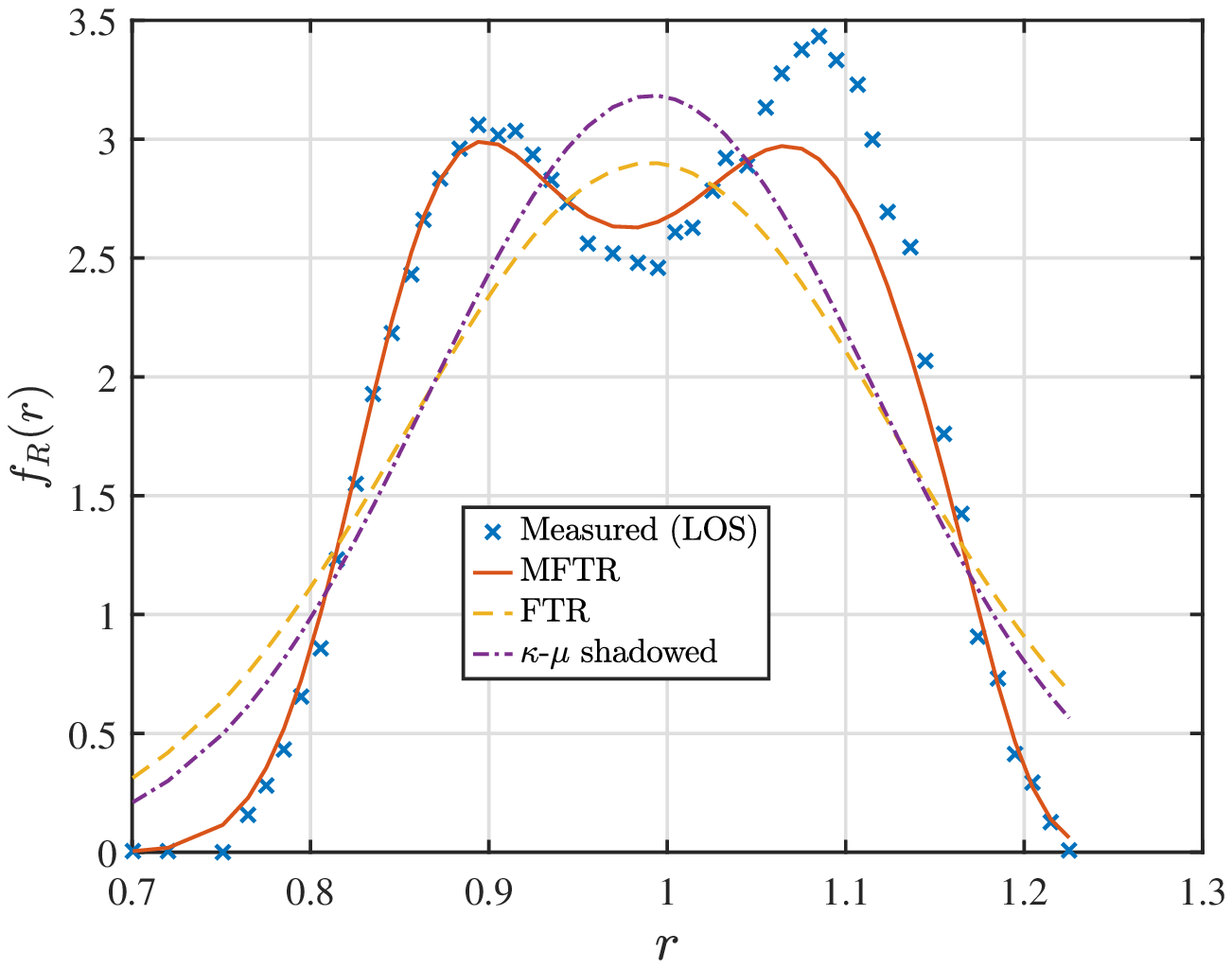}}
%\hspace{-0.6cm}
\subfigure[\hspace{-4cm}]{}\hspace{-0.2in}{\includegraphics[width=0.28\textwidth]{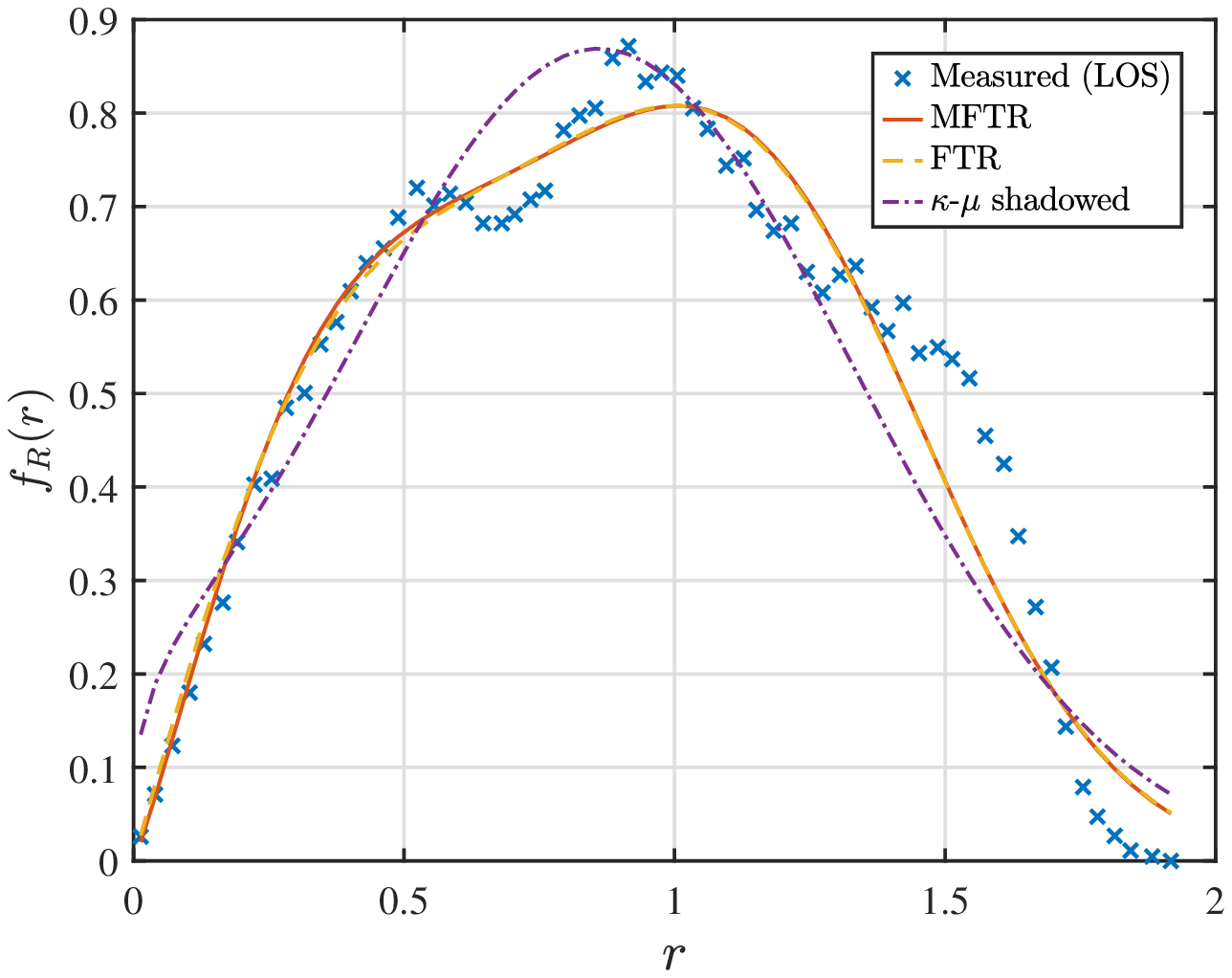}}
\newline
%\hspace{-0.8cm}
\subfigure[\hspace{-4cm}]{}\hspace{-0.4in}
{\includegraphics[width=0.28\textwidth]{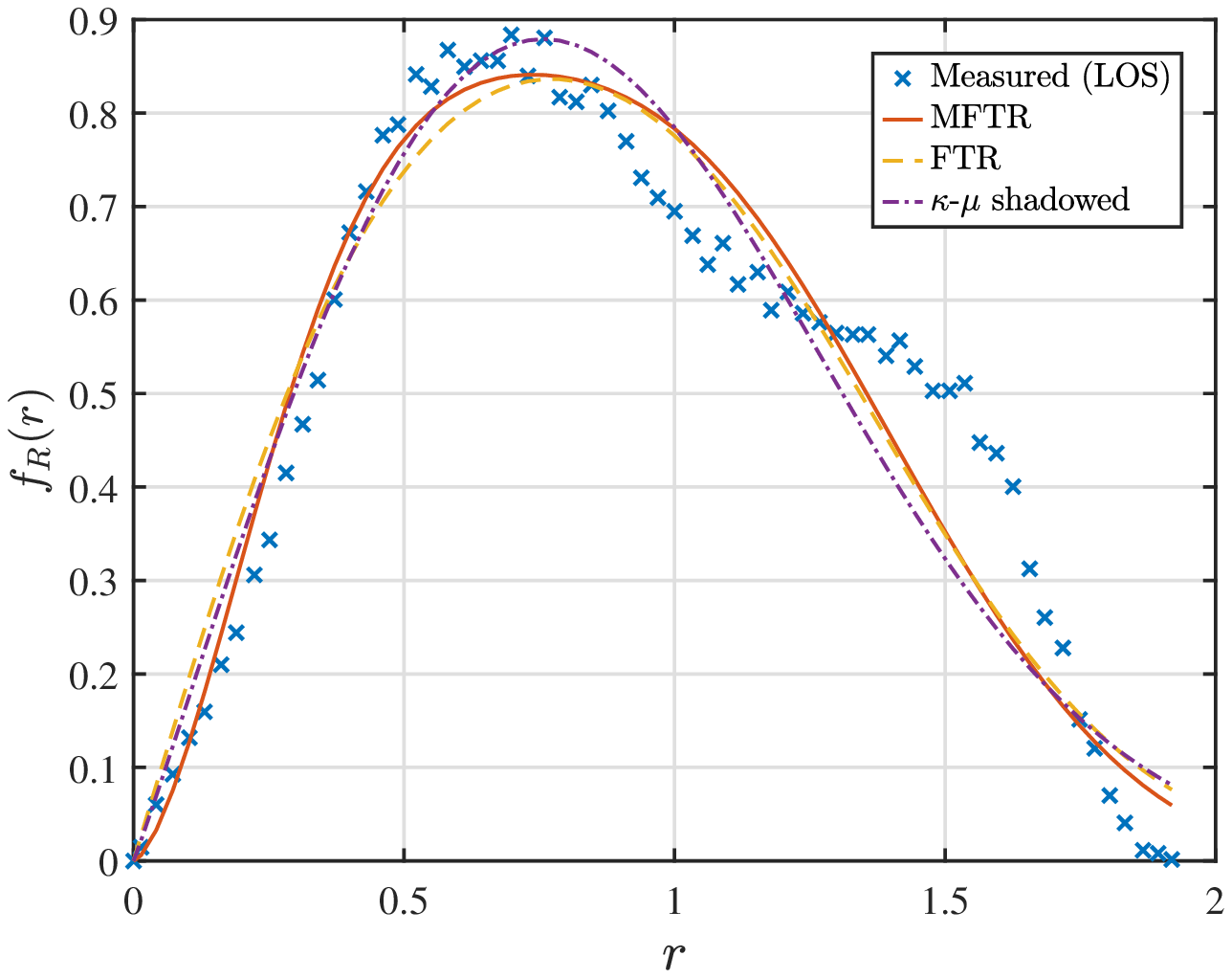}}
%\hspace{-0.6cm}
\subfigure[\hspace{-4cm}]{}\hspace{-0.2in}{\includegraphics[width=0.28\textwidth]{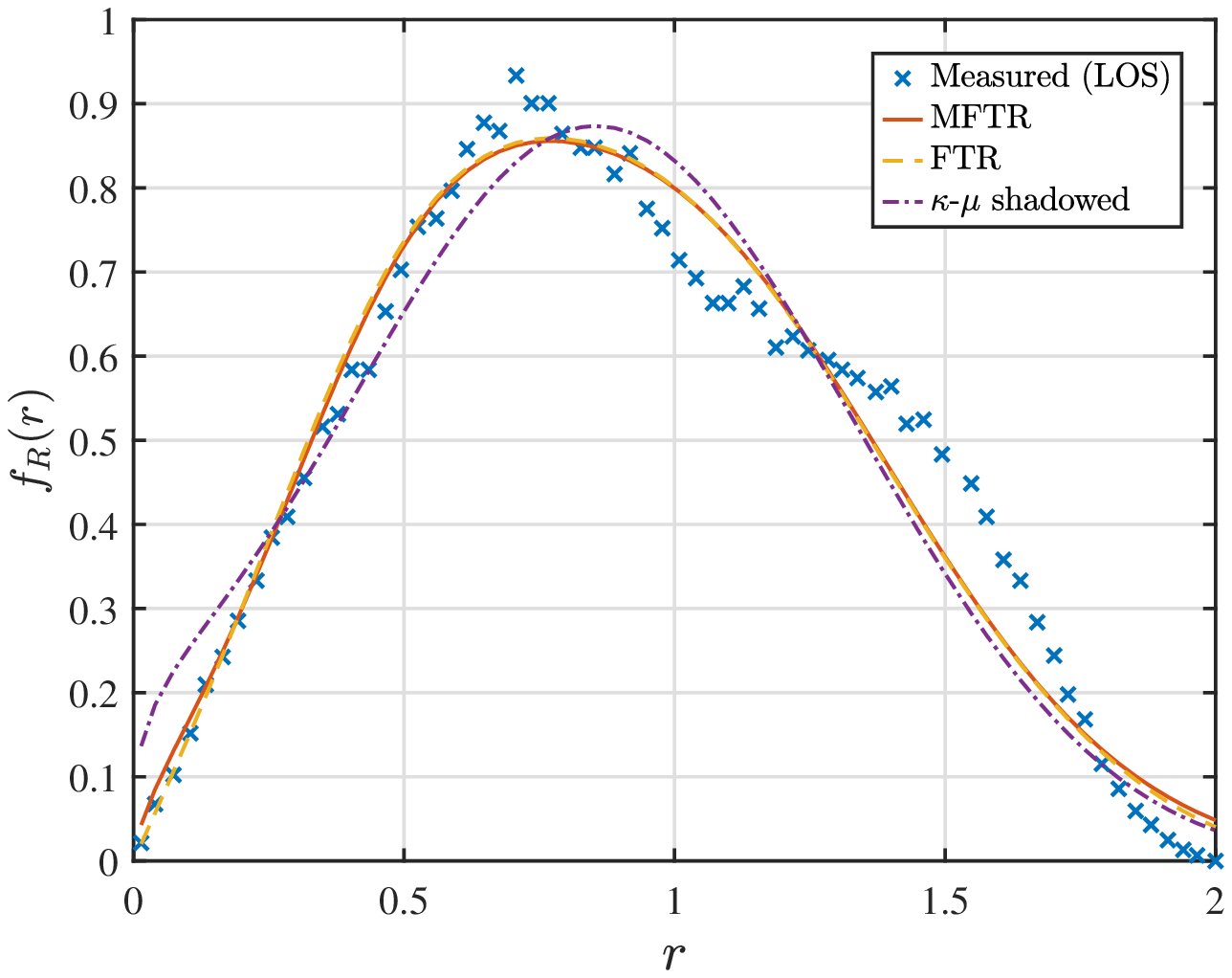}}
%\hspace{-3.1cm}
\newline
\subfigure[\hspace{-4cm}]{}\hspace{-0.4in}
{\includegraphics[width=0.275\textwidth]{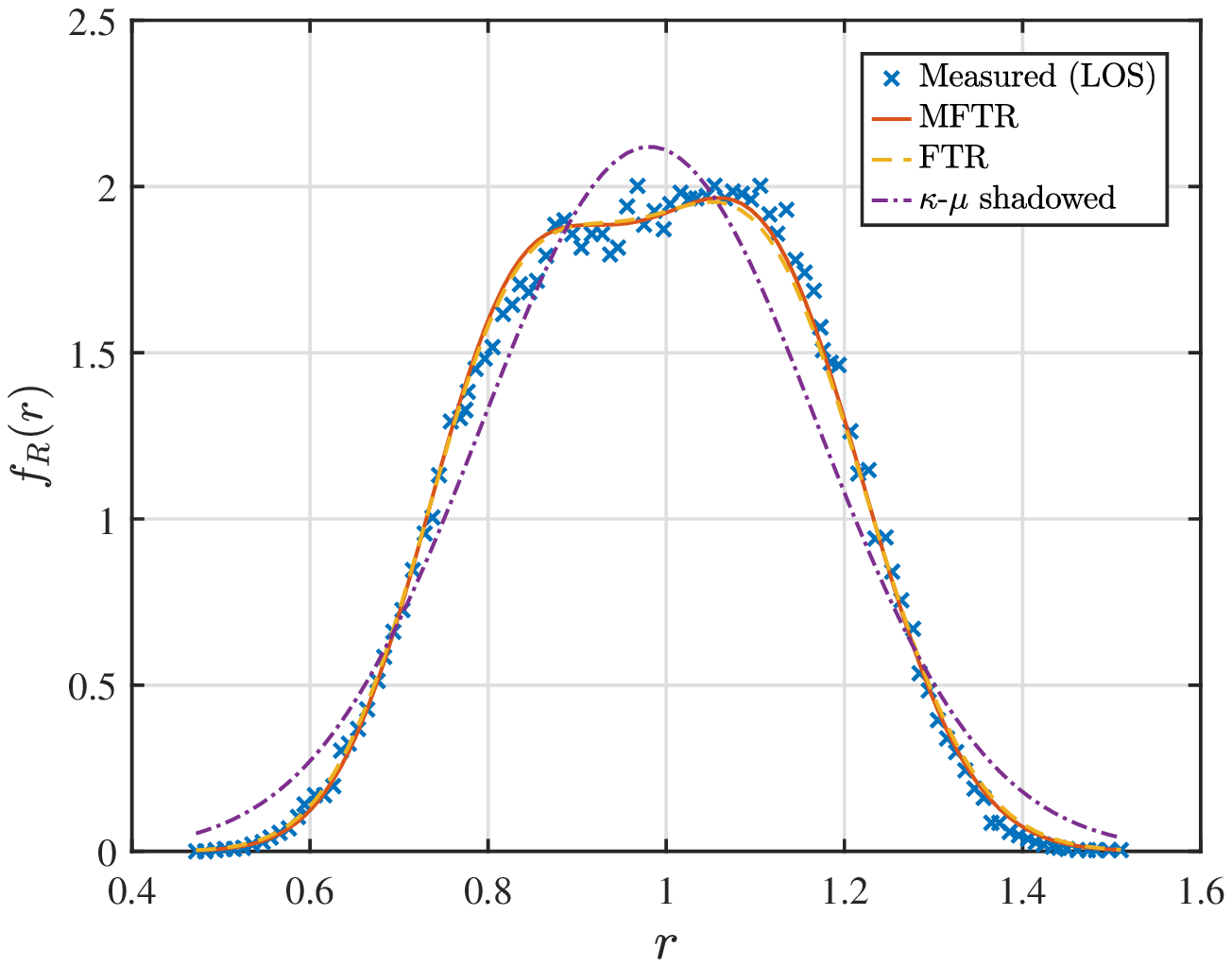}}
%\hspace{-0.6cm}
\subfigure[\hspace{-4cm}]{}\hspace{-0.2in}{\includegraphics[width=0.275\textwidth]{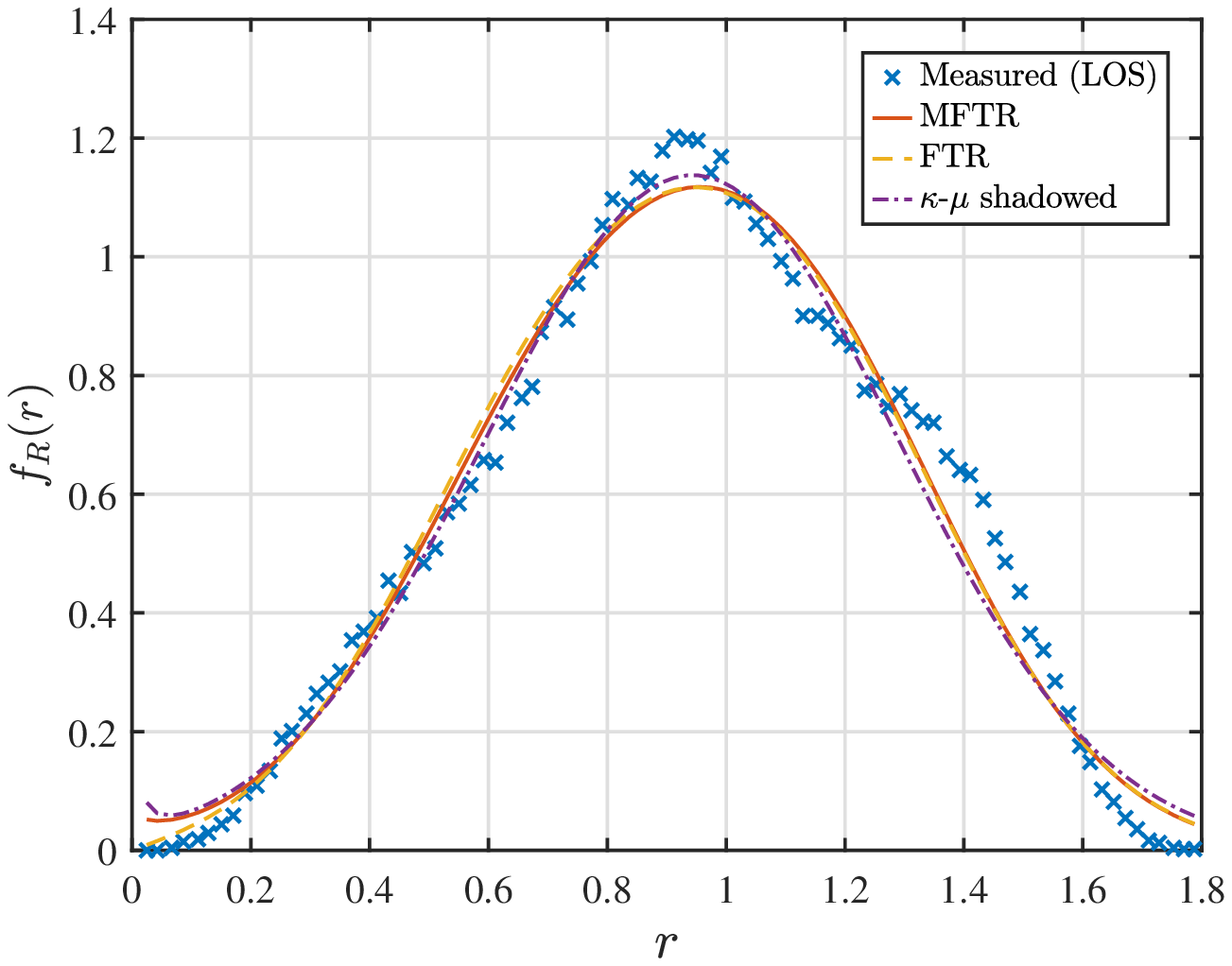}}
\caption{{Empirical vs theoretical PDFs of the received signal amplitude for LOS and NLOS scenarios.}}
\label{fig:fitting}
\end{figure}

{
\section{Empirical Validation}\label{secVal}
In this section, we validate the suitableness of the MFTR model to capture the stochastic features of outdoor THz wireless links in the 142 GHz band. For this purpose, we employ the empirical data presented in \cite{fitting1}, in which sets of experimental measurements have been conducted on the campus of Aalto University in Finland by assuming LOS and NLOS scenarios. More details concerning the experimental setup can be found in \cite{fitting2}. Here, our goal is to assess the ability of the MFTR to capture the
the physical behavior of the THz wireless link by means of multipath clusters in the presence of two dominant specular components with fluctuating amplitudes.}

{
In order to quantify the fitting accuracy, we use the well-defined and widely-accepted mean-square-error (MSE) test. Specifically, the MSE measures the goodness of fit between the empirical and theoretical PDFs
, denoted by $\widehat{f}_R(\cdot)$ and $f_R(\cdot)$, respectively, i.e., 
\begin{align}
\label{eqmse}
\text{MSE}=\frac{1}{T}\sum_{i=1}^{T}\left ( \widehat{f}_R(r_i)-f_R(r_i)\right )^{2},
\end{align}
where $T$ denotes the number of empirical PDF values.}
\begin{table*}[t]
\caption{{Fitting Results of MFTR, FTR and $\kappa$-$\mu$ shadowed fading models for different scenarios.}} 
\label{Tablefit}
\centering
{
\begin{tabular}{|c|c|c|c|c|c|c|c|c|c|c|c|c|c|}
\hline
\textbf{Case}              & \textbf{Distribution}    & $K$    & $\Delta$ & $\mu$  & $m$    & \textbf{MSE}     & \textbf{Case}                        & \textbf{Distribution}    & $K$    & $\Delta$ & $\mu$ & $m$    & \textbf{MSE}     \\ \hline
\multirow{3}{*}{\textbf{\begin{tabular}[c]{@{}c@{}}Fig. \\ 5a\end{tabular}}} & MFTR                     & 10.788 & 0.29     & 39.991 & 90.252 & \textbf{0.06746} & \multirow{3}{*}{\textbf{\begin{tabular}[c]{@{}c@{}}Fig. \\ 5d\end{tabular}}} & MFTR                     & 10.558 & 0.850    & 0.827 & 4.356  & \textbf{0.00304} \\ \cline{2-7} \cline{9-14}   & FTR                      & 38.141 & 0.003    & 1      & 79.323 & 0.32746          &                                                                              & FTR                      & 7.192  & 0.872    & 1     & 4.442  & 0.00306          \\ \cline{2-7} \cline{9-14}   & $\kappa$-$\mu$ shadowed & 11.102 & 0        & 2.986  & 89.337 & 0.32772          &                                    & $\kappa$-$\mu$ shadowed & 4.820  & 0        & 0.645 & 2.288  & 0.00559          \\ \hline
\multirow{3}{*}{\textbf{\begin{tabular}[c]{@{}c@{}}Fig. \\ 5b\end{tabular}}} & MFTR                     & 4.225  & 0.999    & 1.055  & 38.868 & \textbf{0.00351} & \multirow{3}{*}{\textbf{\begin{tabular}[c]{@{}c@{}}Fig. \\ 5e\end{tabular}}} & MFTR                     & 20.717 & 0.403    & 3.108 & 84.333 & \textbf{0.00241} \\ \cline{2-7} \cline{9-14}   & FTR                      & 4.113  & 0.993    & 1      & 89.999 & 0.00354          &                                                                              & FTR                      & 67.293 & 0.390    & 1     & 75.187 & 0.00276          \\ \cline{2-7} \cline{9-14}   & $\kappa$-$\mu$ shadowed & 3.924  & 0        & 0.650  & 2.868  & 0.00791          &                                                                              & $\kappa$-$\mu$ shadowed & 24.834 & 0        & 0.598 & 73.775 & 0.02689          \\ \hline
\multirow{3}{*}{\textbf{\begin{tabular}[c]{@{}c@{}}Fig. \\ 5c\end{tabular}}} & MFTR                     & 3.284  & 0.999    & 1.267  & 5.481  & \textbf{0.00543} & \multirow{3}{*}{\textbf{\begin{tabular}[c]{@{}c@{}}Fig. \\ 5f\end{tabular}}} & MFTR                     & 17.938 & 0.596    & 0.400 & 79.519 & \textbf{0.00347} \\ \cline{2-7} \cline{9-14}   & FTR                      & 1.958  & 1        & 1      & 18.909 & 0.00634          &                                                                              & FTR                      & 6.204  & 0.657    & 1     & 95.614 & 0.00362          \\ \cline{2-7} \cline{9-14}   & $\kappa$-$\mu$ shadowed & 0.525  & 0        & 1.023  & 21.989 & 0.00710          &                                                                              & $\kappa$-$\mu$ shadowed & 39.995 & 0        & 0.107 & 78.744 & 0.00392          \\ \hline
\end{tabular}}
\end{table*}
{
Table~\ref{Tablefit} reports the values of the MSE and the estimated fading parameters using the \emph{fminsearch} function of the Optimization Toolbox of MATLAB for each target distribution. Bold-faced numbers highlight the best-fitting result in each scenario. Here, according to the MSE criteria, we can see that the MFRT fading channel model has achieved the best fittings in all scenarios under study. On the other hand, in Fig.~\ref{fig:fitting}, we compare the theoretical PDFs of MFTR, FTR, and $\kappa$-$\mu$ shadowed fading models against the THz channel measurements for LOS and NLOS environments described in~\cite[Fig.~1]{fitting1}. From all traces, it can be observed that the MFRT model yields a more accurate fit to the empirical distributions. This is in accordance with the MSE values shown in Table~\ref{Tablefit}. At this point, an important remark is in order. Firstly, we want by no means to provide a thorough and tedious comparison between the proposed MFTR model and existing fading models when applied to fit channel measurements in THz or millimeter bands. Instead, we mainly want to provide some evidence that, for selected channel measurement campaigns, the proposed MFTR fading model offers a good balance between versatility, flexibility, the goodness of fit and physical underpinnings. For instance, for the NLOS environment in Fig.~\ref{fig:fitting}a, the MFTR captures the richness of the dispersion through a high value of $\mu$ (i.e., $39.99$), which also allows capturing the pronounced bimodality of the distribution. Another key interpretation of this fading parameter is that as $\mu$ increases, the MFTR model smooths out the goodness of fit as the distribution moves away from the origin on the x-axis. In particular, this behavior is present in Figs.~\ref{fig:fitting}a and~\ref{fig:fitting}e, where the empirical distributions start at approximately $0.7$ and $0.5$ on the x-axis, respectively. Conversely, in all other figures, the empirical measurements begin at $0$ on the x-axis, which translates to small values of $\mu$. In short, the MFTR model provides additional flexibility (i.e., more degrees of freedom) to characterize small-scale empirical fading data compared to $\kappa$-$\mu$ shadowed and FTR models from which it originates.
}

% \usepackage{multirow}
% \usepackage{multirow}

% \usepackage{multirow}

%%%%%%%%%%%%%%%%%%%%%%%%%%%%%%%%
%PERFORMANCE METRICS
\section{Performance Analysis in Wireless Systems}\label{sec5}
In this section, we illustrate the flexibility of the MFTR model when used for performance analysis. For exemplary purposes, we analyze key performance metrics such as the outage probability (OP), both in exact and asymptotic forms, as well as the AoF \cite{Simon}.% as benchmarking metrics often used to assess the performance of wireless systems when communicating over fading channels. 
\subsection{Outage Probability}
The instantaneous Shannon channel capacity per unit bandwidth is defined as
\begin{align}
\label{eq22}
    C=\log_2(1+\gamma).
\end{align}
The outage probability is defined as the probability that the capacity $C$ falls below a certain
threshold rate $R_{\rm th}$, i.e.,  
\begin{align}
\label{eq23}
 P_{\rm out}&=\Pr\left \{ C<R_{\rm th} \right \}=\Pr\left \{ \log_2 (1+\gamma)<R_{th} \right \} \nonumber \\ & = \Pr\left \{ \gamma <\underbrace{2^{R_{\rm th}}-1}_{\gamma_{\rm th}} \right \}.
\end{align}
{Consequently,} the OP is given in terms of SNR's CDF as 
\begin{align}
\label{eq24}
 P_{\rm out}&=F_\gamma\left ( 2^{R_{\rm th}}-1 \right ),
\end{align}
in which $F_\gamma\left ( \cdot \right )$
is given by either \eqref{eq14}, \eqref{eq16} or \eqref{eq19}. Although expressed in exact form in terms of the MFTR CDFs previously derived, $ P_{\rm out}$  provides a limited insight concerning the effect of fading parameters on the system performance. Thus, we introduce a closed-form asymptotic OP expression to evaluate the high-SNR regime's system performance below.

\subsection{Asymptotic Outage \textcolor{black}{ Probability}}
Here, to get further insights about the role of the fading parameters on system performance, we derive an asymptotic expression to investigate the behavior of the OP given in~\eqref{eq24} in the high-SNR regime. Our goal is to obtain an asymptotic expression in the form $P_{\rm out}\simeq \mathrm{G}_c\left(\gamma_{\rm th}/\overline{\gamma}\right)^{\mathrm{G}_d}$~\cite{Giannakis2003}, where $\mathrm{G}_c$ and $\mathrm{G}_d$ denote the power offset and the diversity order, respectively. Hence, the corresponding asymptotic OP is given in the following Proposition.

\begin{prop}\label{propo1}
The asymptotic expression of the OP over MFTR channels can be obtained as
\end{prop}
\begin{align}
\label{eq25}
P_{\rm out}\simeq&\frac{\mu^{\mu-1}(1+K)^{\mu}m^m (2^{R_{\rm th}}-1)^{\mu}}{\Gamma(\mu) \overline{\gamma}^{\mu} \left( m-\left( \Delta-1 \right)K \mu\right)^{m}} \nonumber \\ & \times {}_{2}F_1 \left( 1/2,m;1;\frac{2\Delta K \mu}{K\mu\left( 1-\Delta \right)+m} \right).
\end{align}
\begin{proof}
	See Appendix~\ref{AsympCDF}.
\end{proof}
From~\eqref{eq25}, notice that the diversity order is linked to the number of multipath waves clusters, i.e., $\mathrm{G}_d=\mu$.
\subsection{Amount of Fading}
The AoF is a popular metric used to quantify the severity of fading experienced during transmission under fading channels. It is defined as \cite{Simon}
\begin{align}
\label{eq26}
\rm{AoF}=\frac{\mathbb{V}\left \{\gamma \right \}}{\mathbb{E}\left \{ \gamma \right \}^2}=\frac{\mathbb{E}\left \{ \gamma^2 \right \}-\mathbb{E}\left \{ \gamma \right \}^2}{\mathbb{E}\left \{ \gamma \right \}^2}=\frac{\mathbb{E}\left \{ \gamma^2 \right \}}{\mathbb{E}\left \{ \gamma \right \}^2}-1.
\end{align}
Based on \eqref{eq26}, a closed-form expression of the AoF is given as stated in the following Proposition. 
\begin{prop}\label{propo2}
Assuming non-constrained arbitrary fading values, a closed-form expression for the AoF over MFTR channels can be formulated as
\end{prop}
\begin{align}
\label{eq27}
\rm{AoF}=&\left( 1-\frac{K^2}{(1+K)^2} \right)\left( 1+\frac{1}{\mu} \right)+\frac{K^2}{(1+K)^2} \nonumber \\ & \times \left( 1+\frac{1}{m} \right)\left( 1+\frac{\Delta^2}{2} \right)-1.
\end{align}
\begin{proof}
	See Appendix~\ref{AoF}.
\end{proof}
%%%%%%%%%%%%%  ABEP
\begin{figure*}[t!]
	{
	%\hrulefill
	\small
%	\begin{normalsize}
\begin{align}\label{ABERexact}
 \overline{P_{e}}=&\frac{ \left ( 1+K \right )^{\mu}\mu ^{\mu}}{2^{m-1}\Gamma(\mu+1) \overline{\gamma}^\mu}\left ( \frac{m}{\sqrt{\left ( m+\mu K \right )^2-\mu^2 K^2 \Delta^2}} \right )^m\sum_{q=0}^{  \left \lfloor \frac{m-1}{2} \right \rfloor}\left( -1 \right)^q C_q^{m-1} \left ( \frac{m+\mu K}{\sqrt{\left ( m+\mu K \right )^2-\mu^2 K^2 \Delta^2}} \right )^{m-1-2q}  \nonumber \\ \times &\ \sum_{r=1}^{R}\alpha_r \frac{2^{\mu-1}}{\beta_r^{\mu}\sqrt{\pi}}\Gamma\left ( \mu+1/2 \right ) F_D^{(4)}\Biggl(\mu+\frac{1}{2},1+2q-m,m-q-\frac{1}{2},m-q-\frac{1}{2},\mu-m;\mu+1;-\frac{2\left( 1+K \right)m\mu }{\beta_r\left( m+\mu K \right)\overline{\gamma }}, \nonumber \\ & -\frac{2\left( 1+K \right)m\mu }{\beta_r\left( m+\mu K\left( 1+\Delta  \right) \right)\overline{\gamma }}, 
 -\frac{2\left( 1+K \right)m\mu }{\beta_r\left( m+\mu K\left( 1-\Delta  \right) \right)\overline{\gamma }},-\frac{2\left( 1+K \right)\mu }{\beta_r\overline{\gamma }} \Biggl).
 \end{align}
%	\end{normalsize} 
	\hrulefill
	}
\end{figure*}
{
\subsection{Average BER}
The exact solution for the average BER affected by additive white Gaussian
noise (AWGN), over the output SNR can be defined as in~\cite{Simon} by
\begin{equation}\label{abep1}
\overline{P_{e}}=\int_{0}^{\infty}P_{E}(x)f_\gamma(x)dx,
\end{equation}
where $P_{E}(x)$ denotes the conditional  error probability (CEP). Considering integration by parts in \eqref{abep1}, the average BER can be computed as function of the CDF, given by
\begin{equation}\label{abep2}
\overline{P_{e}}=-\int_{0}^{\infty} P_{E}^{'}(x)F_\gamma(x)dx,
\end{equation}
where $P_{E}^{'}(x)$ denotes the first order derivative of the CEP. For several modulation schemes, $P_{E}(x)$ can be defined as in~\cite{Simon} by
\begin{equation}\label{abep3}
P_{E}(x)=\sum_{r=1}^{R}\alpha_rQ\left ( \sqrt{\beta_rx} \right ),
\end{equation}
where $\left \{\beta_r, \alpha_r  \right \}_{r=1}^{R}$ are constants that depend on the type of modulation. Taking the derivative of \eqref{abep3}, it follows that
\begin{equation}\label{abep4}
P_{E}^{'}(x)=-\sum_{r=1}^{R}\alpha_r\sqrt{\frac{\beta_r}{8\pi x}}e^{-\frac{\beta_r x}{2}}.
\end{equation}
Finally, by substituting~\eqref{eq14} and \eqref{abep4} into \eqref{abep2} followed by some manipulations and then making use of~\cite[~Eq.~(43)]{Brychkov}, a closed-form fashion of the average BER can be expressed as in \eqref{ABERexact}.
}
{
\subsection{Asymptotic Average BER}
In this section, we derive an asymptotic closed-form expression for the average BER in order to gain more insights into the impact of the fading parameters of the systems. Here, we consider the behavior in the high SNR regime where $\overline{\gamma}\rightarrow \infty$. Again, our aim is to express the asymptotic average BER as $\overline{P_{e}}^{\infty}\simeq \mathrm{G}_c\overline{\gamma}^{-\mathrm{G}_d}$~\cite{Giannakis2003}. For that purpose, we use the approach in~\cite[Proposition 3]{Giannakis2003}, where asymptotic $\overline{P_{e}}$ can be obtained from the MGF given in~(\ref{eq9}). In light of this, we first compute $\abs{ {{{\mathcal{M}}}_{\gamma }}\left( s \right)}$ as 
\begin{align}
\label{abepAsin1}  
 \abs{  {{{\mathcal{M}}}_{\gamma }}\left( s \right)}=&\frac{{{m}^{m}}{{\mu }^{\mu }}{{\left( 1+K \right)}^{\mu }}{{}}}{{{\left( \sqrt{\left( m+\mu K \right)^{2}-\left( \mu K\Delta  \right)^{2}} \right)}^{m}}}\nonumber \\ & \times {{P}_{m-1}}\left( \frac{m\mu \left( 1+K \right)}{\sqrt{\left( m+\mu K \right)^{2}-\left( \mu K\Delta  \right)^{2}}} \right)\frac{\abs{s}^{-1}}{\overline{\gamma }^{\mu}}
 \nonumber \\ & 
 +o\left ( \abs{s}^{-1} \right ),
  \end{align}
  where we write a function $a(x)$ of $x$ as $o(x)$ if $\lim_{x\to 0} a(x)/x=0 $, in which the Legendre polynomial is computed from~(\ref{eq11}). Then, by carrying out the same methodology given in~\cite[Proposition 1 and 3]{Giannakis2003}, $\overline{P_{e}}^{\infty}$ can be defined as
  \begin{align}
\label{abepAsin2}  
    \overline{P_{e}}^{\infty}\simeq & \frac{{{m}^{m}}{{\mu }^{\mu }}{{\left( 1+K \right)}^{\mu }}{{}}}{{{\left( \sqrt{\left( m+\mu K \right)^{2}-\left( \mu K\Delta  \right)^{2}} \right)}^{m}}}  \sum_{r=1}^{R}\alpha_r \frac{2^{\mu-1}}{\beta_r^{\mu}\sqrt{\pi}}\Gamma\left ( \mu+\tfrac{1}{2} \right ) \nonumber \\ & \times {{P}_{m-1}}\left( \frac{m\mu \left( 1+K \right)}{\sqrt{\left( m+\mu K \right)^{2}-\left( \mu K\Delta  \right)^{2}}} \right)\frac{1}{\overline{\gamma }^\mu}.
\end{align}
As in $P_{\rm out}$ metric, from~\eqref{abepAsin2}, it is clearly seen that the diversity order of $ \overline{P_{e}}$ is associated to the number of multipath waves clusters, i.e., $\mathrm{G}_d=\mu$.
}
{
\subsection{Ergodic Capacity}
The EC is defined as the maximum achievable rate averaged over all the fading contributions. Mathematically speaking, the EC can be expressed as~\cite{Simon}
\begin{align}
\label{ECeq1}  
    \overline{C}= &\int_{0}^{\infty}\log_2 \left ( 1+x \right )f_\gamma(x)dx,
\end{align}
where $f_\gamma(\cdot)$ is the PDF of the received SNR $\gamma$ under MFTR fading channels. For the sake of mathematical tractability, we employ the alternative PDF of the MFTR model given in terms of the mixture of Gamma Distributions to compute $\overline{C}$. Hence, substituting~(\ref{eq15}) into \eqref{ECeq1}, and after some manipulations, EC can be obtained  as 
\begin{align}
\label{ECeq2}  
    \overline{C}= &\sum_{i=0}^{\infty}{w_i}\frac{\left ( \frac{\lambda}{\nu} \right )^{\lambda}}{\Gamma(\lambda)\ln (2)}\sum_{k=1}^{\infty}\frac{1}{k}\Gamma(k+\lambda)\bold{U}\left ( k+\lambda,1+\lambda,\frac{\lambda}{\nu} \right ),
\end{align}
where $w_i$ is given by~(\ref{eqWeights}), $\lambda=\mu+i$, and $\nu=\frac{\overline{\gamma}(\mu+i)}{\mu(K+1)}$.
}

\begin{figure}[t]
\centering
\psfrag{A}[Bc][Bc][0.7]{$\mathrm{\textit{m}=10}$}
\psfrag{B}[Bc][Bc][0.7]{$\mathrm{\textit{m}=2}$}
\includegraphics[width=1\linewidth]{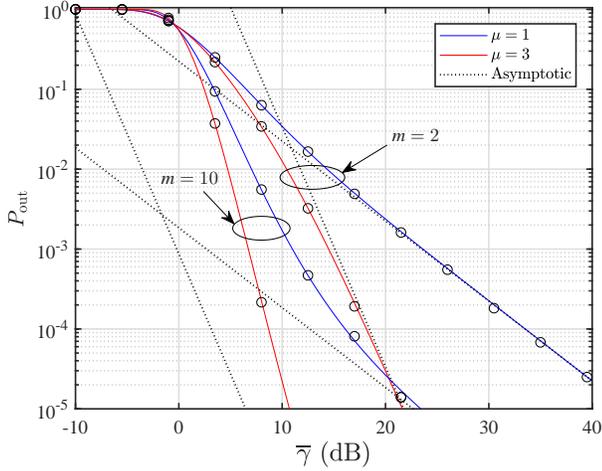}
\caption{ $ P_{\rm out}$ as a function of the average SNR, for different values of $\mu$ and $m$. The remaining fading parameters are: $K=15, \ \Delta=0.1, \ R_{\rm th}=1$. \textcolor{black}{Solid lines correspond to the  $P_{\rm out}$ derived from \eqref{eq24}, and dotted lines correspond to the asymptotic $ P_{\rm out}$ derived from \eqref{eq25}}. Markers denote MC simulations  }
\label{figOP1}
\end{figure}

%%%%%%%%%%%%%%%%%%%%%%%%%%%%%%%%
%NUMERICAL RESULTS
\section{Numerical Results}\label{sec6}
In this section, we provide illustrative numerical results along with MC simulations{\footnote{{\textbf{Reproducible Research:} The simulation code on how to generate the PDF and CDF of the MFTR model in both amplitude and power/SNR distributions, as well as the Monte Carlo simulation, is available at: \url{https://github.com/JoseDavidVega/MFTR-Fading-Channel-Model}}}} to verify the analytical performance metrics derived in the previous section by assuming different fading conditions. 

Figs.~\ref{figOP1} and.~\ref{figOP2} illustrate the impact of different propagation mechanisms, namely, NLOS- LOS-condition, LOS fluctuation, the existence of one/two dominant components in cluster 1, and clustering of multipath waves, on the OP performance. In both figures, we assume $K=15$ and $R_{\rm th}=1$. Specifically, in Fig.~\ref{figOP1}, we evaluate the OP vs. the average SNR by varying $\mu$ and $m$ for $\Delta=0.1$. From all traces, it can be observed that a contribution of the LOS component with a mild fluctuation ($m=10$) together with a rich scattering environment ($\mu=3$) favors the OP performance. Conversely, when both the LOS fluctuation is severe (i.e., lower values of $m$), and a poor scattering condition exists ($\mu=1$), the OP performance is noticeably reduced. On the other hand, in Fig.~\ref{figOP2}, we show the OP as a function of the average SNR for different fading values of $\Delta$ and $\mu$, assuming a mild fluctuation ($m=5$) for the LOS component. Here, we can see that the combination of similar ($\Delta=0.9$) specular components with many clusters of multipath waves ($\mu=4$) derives into a better OP behavior. However, in the opposite scenario, i.e., dissimilar ($\Delta=0.1$) specular components together with reduced clustering of scattering waves, the OP significantly deteriorates.

\begin{figure}[t]
\centering
\psfrag{A}[Bc][Bc][0.7]{$\Delta=0.1$}
\psfrag{B}[Bc][Bc][0.7]{$\Delta=0.9$}
\includegraphics[width=1\linewidth]{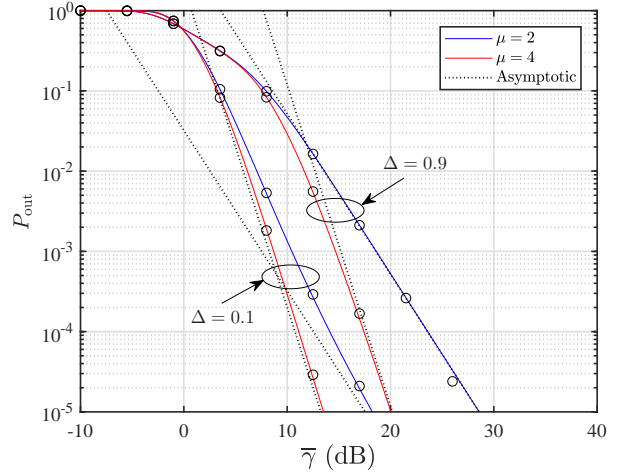}
\caption{ $ P_{\rm out}$ as a function of the average SNR, for different values of $\mu$ and $\Delta$. The remaining parameters are: $K=15, \ m=5, \ R_{\rm th}=1 $. \textcolor{black}{Solid lines correspond to the  $P_{\rm out}$ derived from \eqref{eq24}, and dotted lines correspond to the asymptotic $ P_{\rm out}$ derived from \eqref{eq25}}. Markers denote MC simulations }
\label{figOP2}
\end{figure}

%\begin{figure}[t]
%\centering
%\psfrag{Z}[Bc][Bc][0.6]{$\mathrm{\textit{n}=256}$}
%\includegraphics[width=1\linewidth]{./Figures/AoF.eps}
%\caption{The AoF vs. $\Delta$ over MTFR fading channels. \textcolor{black}{Solid lines correspond to the AoF computed via \eqref{eq27} and markers denote Monte Carlo simulations.}}
%\label{figAoF}
%\end{figure}

Furthermore, in Figs.~\ref{figOP1} and.~\ref{figOP2}, we see that the number of clusters of multipath waves contributes directly to the OP slope. This means that the decay in the OP is steeper (i.e., better performance) as the number of scattering wave clusters increases. This is in coherence with the derived diversity order, i.e., $\mathrm{G}_d=\mu$.

%Fig.~\ref{figAoF} shows the effect of the AoF experienced in MFTR fading channels for different fading values, as a function of $\Delta$. It becomes clear that the largest AoF occurs for low values of the $m$, $K$, and $\mu$ parameters indicating that the channel is subject to a more severe multipath fading. Moreover, the AoF increases as $\Delta$ becomes large, since partial cancellation between the dominant components becomes more likely as their amplitudes are more balanced.

{Finally, Fig.~\ref{figAbep} depicts the average BER for binary phase-shift keying (BPSK) modulation vs. the average SNR by varying $\mu$ and $\Delta$ for a fixed LOS fluctuation, i.e., $m=2$. The parameters for BPSK modulation in \eqref{ABERexact} are set to $R=1$, and $\left \{\beta_r, \alpha_r  \right \}_{r=1}^{R}=\left \{2,1\right \}$. From all curves, we can see that the eventual similarity specular components in the first cluster does not dominate ($\Delta=0.1$), experiencing a moderate fluctuation ($m=2$) in a rich scattering condition ($\mu=3$) yields the best $\overline{P_{e}}$ behavior. Conversely, average BER performance worsens as $\Delta$ increases coupled with a dispersion-poor environment. Also, as in the $P_{\rm out}$ metric, the $\overline{P_{e}}$ slope directly depends on the number of clusters of multipath waves, which is associated with the obtained diversity order $\mathrm{G}_d=\mu$.}
%. On the other hand, choosing large values for $m$, $K$, and $\mu$, it is observed that the AoF deteriorates, improving the channel conditions.

%%%%%%%%%%%%%%%%%%%%%%%%%%%%%%%%
%CONCLUSIONS

\section{Conclusions}\label{sec7}
We introduced, for the first time in the literature, a stochastic fading model that combines the key features of ray-based and power envelope-based approaches. This newly proposed fading model unifies and generalizes both the FTR and the $\kappa$-$\mu$ shadowed fading models, with a comparable mathematical complexity. The MFTR model captures, through a limited set of physically-meaningful parameters, a number of propagation conditions that appear in many practical scenarios: LOS/NLOS propagation, amplitude imbalances between dominant specular components, random fluctuation of the dominant specular waves, and clustering of multipath waves. In order to facilitate its use for performance analysis purposes, alternative expressions for its statistics are derived, which enable us to express the desired performance metric under MFTR fading directly from those available either for $\kappa$-$\mu$ shadowed or Nakagami-$m$ fading. All these features make the MFTR model a strong contender to being the most generalized and comprehensive fading channel model in the state-of-the-art up to the present date.
\begin{figure}[t]
\centering
%\psfrag{Z}[Bc][Bc][0.6]
\psfrag{A}[Bc][Bc][0.7]{$\Delta=0.1$}
\psfrag{B}[Bc][Bc][0.7]{$\Delta=0.9$}%{$\mathrm{\textit{n}=256}$}
\includegraphics[width=1\linewidth]{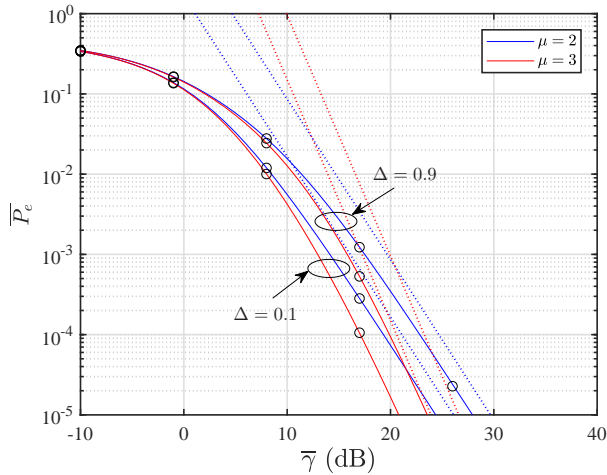}
\caption{{$ \overline{P_{e}}$ for BPSK modulation as a function of the average SNR, for different values of $\mu$ and $\Delta$. The remaining parameters are: $K=10, \ m=2 $. {Solid lines correspond to the $\overline{P_{e}}$ derived from \eqref{ABERexact}, and dotted lines correspond to the asymptotic $\overline{P_{e}}^{\infty}$ derived from \eqref{abepAsin2}}. Markers denote MC simulations.}}
\label{figAbep}
\end{figure}

\appendices
%\renewcommand\thechapter{\Roman{chapter}}

%%%%%%%%%%%%%%%%%%%%%%%%%%%%%%%%
%%%%%%%%%%%%%%%%%%%%%%%%%%%%%%%%
%APPENDIX A
\section{ Proof of Lemma~\ref{lemma1}}

\label{ap:MGF}
Departing from a generalized physical cluster-based  model \cite{kms,kappamuYacub}, let us consider a compact form expression for the signal power, {$W$},  {as in \eqref{eqclusters},} given by
\begin{align}
\label{Lema1eq1}
  {W}=\sum_{i=1}^{\mu}\left | Z_i+\sqrt{\zeta}  p_i  \right |^2.
\end{align}
%where $\mu$ is the number of clusters, $Z_i= {X}_i+j{Y}_i$ with ${X}_i$ and ${Y}_i$ being independent Gaussian RVs with zero-mean and variance $\sigma^2$, i.e., ${X}_i,{Y}_i \sim \mathcal{N}(0,\sigma^2)$, representing the $i$th cluster of scattering waves, and ${\zeta}$ denotes the fluctuation of the LOS component, which is statistically distributed as a unit-mean gamma RV with shape parameter $m$. On the other hand, 
In purely cluster-based models, the total power of the dominant components can be distributed indistinctly throughout the clusters. Here, without loss of generality, we assume that two specular components are allocated to cluster 1, and each of the remaining clusters may include one specular component. With this in mind, the dominant specular component, denoted by $p_i$, of the $i$th cluster can be expressed as
 \begin{equation}
 \label{Lema1eq2}
         p_i=\begin{cases}
       U_ie^{j\varphi_i}, & \text{for $\forall {i} >1$} \\
       V_1e^{j\phi_1}+ V_2e^{j\phi_2}=e^{j\phi_1}\left ( V_1+ V_2e^{j\alpha} \right ), & \text{for $ {i}=1$},
         \end{cases}
\end{equation}
%in which $V_n \exp(j\phi_n)$ for $n=\left \{1,2  \right \}$ denote the $n$th specular component with constant amplitude $V_n$ and a uniformly distributed random phase, such that, $\phi_n \sim \mathcal{U}(0, 2\pi)$. Moreover,
where $\alpha=\phi_2-\phi_1$ is the phase difference between the specular components of cluster 1. Given that $\phi_n \sim \mathcal{U}(0, 2\pi)$ and because of the modulo $2\pi$ operation, it follows that, $\alpha\sim \mathcal{U}(0, 2\pi)$ \cite{Rao2015}. 

Let us consider the channel model given in \eqref{Lema1eq1} conditioned to a particular realization $\alpha=\theta$ of the RV characterizing the phase difference between the two LOS components in cluster 1. Then, \eqref{Lema1eq1} can be seen as a $\kappa$-$\mu$ shadowed distributed RV for a given $\theta$, where $\abs{p_1}$ is now constant (i.e., deterministic) and no longer arbitrarily distributed. Based on this, the mean power of the dominant components for the underlying $\kappa$-$\mu$ shadowed RV in \eqref{Lema1eq1}, conditioned on $\theta$, i.e., $d_\theta^2 \stackrel{\Delta}{=} d^2|_{\alpha=\theta}$ is given by
\begin{align}
\label{Lema1eq3}
 d_\theta^2 =&  \sum^\mu_{i=1} \abs{p_i}^2=\abs{V_1+ V_2e^{j\theta} }^2 + \sum^\mu_{i=2} U^2_i
 \nonumber \\ =&
{\abs{V_1+ V_2(\cos{\theta}+j\sin{\theta})}^2  + \sum^\mu_{i=2} U^2_i}
  \nonumber \\ =&
{ (V_1+V_2\cos{\theta})^2+(V_2\sin{\theta})^2  + \sum^\mu_{i=2} U^2_i}
\nonumber \\ =&
{ 
 V_1^2+2V_1V_2 \cos{\theta}+V_2^2 (\cos^2{\theta}+\sin^2{\theta})+ \sum^\mu_{i=2} U^2_i}
\nonumber \\ =&
 V_1^2+V_2^2+2V_1V_2\cos{\theta}+ \sum^\mu_{i=2} U^2_i.
\end{align}
With the help of \eqref{Lema1eq3}, the ratio between the total power of the dominant components and the total power of the scattered waves, conditioned on $\theta$ can be formulated as
\begin{align}
\label{Lema1eq4}
 \kappa_\theta = \frac{d_\theta^2}{2\sigma^2 \mu}.
\end{align}
The conditional average SNR for the fading model described in \eqref{Lema1eq1} will be
\begin{align}
\label{Lema1eq5}
 \overline{\gamma}_\theta =&{\frac{E_s}{N_0}} \mathbb{E}\left \{ {W_{\theta}}\right \}
\nonumber\\=&{\frac{E_s}{N_0}}\left(V_1^2+V_2^2+2V_1V_2\cos{\theta}+ \sum^\mu_{i=2} U^2_i+2\sigma^2\mu\right) 
\nonumber \\ =& {\frac{E_s}{N_0}}2\sigma^2\mu(1+\kappa_\theta).
\end{align}
{where $W_{\theta}$ denotes the signal power in \eqref{Lema1eq1}, in which $\abs{p_1}$ is conditioned on $\theta$.} 
%Similarly, from \eqref{eq5} we can compute the mean power of the dominant components for the MFTR model, conditioned on $\theta$. Thus, we have
%\begin{align}
%\label{Lema1eq6}
% d_{\theta,{\rm MFTR}}^2= &{{\left( {{V}_{1}}\cos {{\phi }_{1}}+{{V}_{2}}\cos {{\phi }_{2}} \right)}^{2}}+{{\left( {{V}_{1}}\sin {{\phi }_{1}}+{{V}_{2}}\sin {{\phi }_{2}} \right)}^{2}}
%\nonumber \\& +\sum\limits_{i = 2}^\mu  {U^2_i \left( {\cos ^2 \varphi _i  + sen^2 \varphi _i } \right)} 
%\nonumber \\ =&  V_1^2+V_2^2+2V_1V_2\cos{\left(\phi_2-\phi_1\right)}+\sum\limits_{i = 2}^\mu  U^2_i
%\nonumber \\ =&V_1^2+V_2^2+2V_1V_2\cos{\theta}+\sum\limits_{i = 2}^\mu  U^2_i.
%\end{align}
%It is evident that the parameters in \eqref{Lema1eq3} and \eqref{Lema1eq6}, related to the $\kappa$-$\mu$ shadowed and MFTR models, respectively, are equivalent. 
Hence, with the above definitions, we are ready to find an insightful connection between the $\kappa$-$\mu$ shadowed and the MFTR models. To accomplish this, by inserting \eqref{eq6}, \eqref{eq7} and \eqref{Lema1eq3}  into \eqref{Lema1eq4}, we get
\begin{align}
\label{Lema1eq7}
  \kappa_\theta = &\frac{V_1^2+V_2^2+2V_1V_2\cos{\theta}+\sum\limits_{i = 2}^\mu  U^2_i}{2\sigma^2 \mu}
	\nonumber \\ = & \frac{\left( V_{1}^{2}+V_{2}^{2} +\sum\limits_{i = 2}^\mu  U^2_i \right)\left( 1+\Delta \cos \theta  \right)}{2{{\sigma }^{2}}\mu }
	\nonumber \\ = & K\left( 1+\Delta \cos \theta  \right).
\end{align}
Note that from {\eqref{eq8}} and \eqref{Lema1eq5} we can write, respectively
\begin{align}
\label{Lema1eq8}
 \frac{1+K}{\overline{\gamma }\ }=\frac{1}{ {\frac{E_s}{N_0}}2\sigma^2\mu},
\end{align}
\begin{align}
\label{Lema1eq9}
 \frac{1+\kappa_\theta}{\overline{\gamma }_\theta\ }=\frac{1}{ {\frac{E_s}{N_0}}2\sigma^2\mu},
\end{align}
and equating \eqref{Lema1eq8} and \eqref{Lema1eq9}, it is clear that
\begin{align}
\label{Lema1eq10}
 \frac{1+\kappa_\theta}{\overline{\gamma }_\theta\ }= \frac{1+K}{\overline{\gamma }\ }.
\end{align}
Here, we point out that with the aid of the key findings in \eqref{Lema1eq7} and \eqref{Lema1eq10}, we can derive the MGF of the MFTR model from the conditional MGF of the $\kappa$-$\mu$ shadowed, which is given by 
\begin{align}
\label{Lema1eq11}
  {{M}^{\kappa \mu s}_{\gamma \left| \theta  \right.}}(s)=\frac{{{\left( 1-\frac{\overline{{{\gamma}}}_{\theta }}{\mu \left( 1+{{\kappa }_{\theta }} \right)}s \right)}^{m-\mu }}}{{{\left( 1-\frac{\mu {{\kappa }_{\theta }}+m}{m}\frac{\overline{{{\gamma }}}_{\theta }}{\mu \left( 1+{{\kappa }_{\theta }} \right)}s \right)}^{m}}}
\end{align}
Now, taking into account \eqref{Lema1eq7} and \eqref{Lema1eq10} into \eqref{Lema1eq11}, yields
\begin{align}
\label{Lema1eq12}
  {{M}^{\kappa \mu s}_{\gamma \left| \theta  \right.}}(s)=\frac{{{m}^{m}}{{\mu }^{\mu }}{{\left( 1+K \right)}^{\mu }}{{\left( \mu \left( 1+K \right)-\overline{\gamma }s \right)}^{m-\mu }}}{\left ( a(s)+b(s) \cos{\theta} \right )^{m}}
\end{align}
where, we have defined
\begin{align}
\label{Lema1eq13}
  a(s)=m\mu \left( 1+K \right)-\left( \mu K+m \right)\overline{\gamma }s, \quad b(s)=-\mu K\Delta \overline{\gamma }s.
\end{align}
The MGF of the SNR of the MFTR model can be obtained by averaging \eqref{Lema1eq12} with respect to the {RV\footnote{Recall that $\frac{1}{2\pi}\int_{0}^{2\pi}f(\theta)d\theta\equiv\frac{1}{\pi}\int_{0}^{\pi}f(\theta)d\theta$ because the integrand is symmetric with respect to $\pi$.}} $\theta$, i.e.,
\begin{align}
\label{Lema1eq14}
 {{M}_{\gamma }}\left( s \right)=\frac{m^m\mu ^\mu  \left( 1+K \right)^\mu }{\left( \mu \left( 1+K \right)-\overline{\gamma }s \right)^{\mu-m }}\frac{1}{\pi}\underset{I_1}{\underbrace{\int_{0}^{\pi}\frac{d\theta}{\left ( a(s)+b(s) \cos{\theta} \right )^{m}}}},
\end{align}
Using~\cite[Eq.~(3.661.4)]{Gradshteyn},~$I_1$ can be evaluated in exact closed-form; thus \eqref{eq9} is obtained, in which $\mathcal{R}(\mu,m,K,\Delta;s)=\left [ a(s) \right ]^2-\left [ b(s) \right ]^{2}$. This completes the proof.

%
%APPENDIX 
%\newpage
\section{Proofs of Lemmas~\ref{lemma2} and~\ref{lemma3} }
\subsection{Proof of Lemma~\ref{lemma2}}\label{ap:PFDCDF}
Notice that the $\mathcal{R}(\mu,m,K,\Delta;s)$ polynomial within the MGF given in \eqref{eq9} can be factored as
\begin{align}
\label{Lema2eq1}
\mathcal{R}&(\mu,m,K,\Delta;s)= \left[ \left( 1+K \right)m\mu -\left( m+\mu K\left( 1+\Delta  \right) \right)\overline{\gamma }s \right]\nonumber \\  &\times \left[ \left( 1+K \right)m\mu -\left( m+\mu K\left( 1-\Delta  \right) \right)\overline{\gamma }s \right].
\end{align}
For the sake of mathematical legibility, we define the following ancillary terms
\begin{align}
\label{Lema2eq2}
 a_1=&\frac{\left( 1+K \right)m\mu }{\left( m+\mu K \right)\overline{\gamma }}, \quad a_2=\frac{\left( 1+K \right)m\mu }{\left( m+\mu K\left( 1+\Delta  \right) \right)\overline{\gamma }}, \nonumber \\ 
 a_3=&\frac{\left( 1+K \right)m\mu }{\left( m+\mu K\left( 1-\Delta  \right) \right)\overline{\gamma }}, \quad   a_4=\frac{\left( 1+K \right)\mu }{\overline{\gamma }}.
\end{align}
Inserting~\eqref{eq11} together with~\eqref{Lema2eq1}-\eqref{Lema2eq2} into~\eqref{eq9}, the MGF of the MFTR model with $m\in \mathbb{Z}^+$ can be rewritten as
\begin{align}
\label{Lema2eq3}
  {{M}_{\gamma }}\left( s \right)=&\frac{\left( a_2 a_3 \right)^{m/2}}{\Gamma\left( \mu \right)2^{m-1}a_4^{m-\mu}} \sum_{q=0}^{  \left \lfloor \frac{m-1}{2} \right \rfloor}\frac{C_q^{m-1}  }{(-1)^{-q}}\left[ \frac{\sqrt{a_2a_3}}{a_1} \right]^{m-1-2q} \nonumber \\ &  \times
  \frac{\Gamma\left( \mu \right)}{s^{\mu}}\left( 1-\frac{a_1}{s} \right)^{m-1-2q}\left( 1-\frac{a_2}{s} \right)^{\frac{1}{2}+q-m}\nonumber \\ &  \times \left( 1-\frac{a_3}{s} \right)^{\frac{1}{2}+q-m}\left( 1-\frac{a_4}{s} \right)^{ m-\mu}.
\end{align}
Finally, the PDF is obtained straightforwardly from this MGF through the inverse Laplace transform (LT), i.e., $f_\gamma(x)=\mathcal{L}^{-1}\left [ M_\gamma(-s);x \right ]$; thus,~\eqref{eq13} arises from~\eqref{Lema2eq3} with the help of the LT pair given in~\cite[Eq.~(4.24.3)]{Srivastava}. Similarly,~\eqref{eq14} is obtained by considering that $F_\gamma(x)=\mathcal{L}^{-1}\left [ M_\gamma(-s)/s;x \right ]$.

%%%%%%%%%%%%%%%%%%%%%%%%%%%%%%%%%%%%%%%%%%%%%%%%%%%%%%%%%%%%%%%  LEMMA 3
\subsection{Proof of Lemma~\ref{lemma3}}\label{ap:PFDCDFreal}
Many distributions in the literature of wireless channel modeling can be expressed in terms of a mixture of Gamma distributions, either in exact \cite{Ermolova2016} or approximate form \cite{mixtureGammas2001}. On this basis, we aim to express the MFTR's statistics as a mixture of Gamma distributions. To do this, we start from the exact expression of the $\kappa$-$\mu$ shadowed model as an infinite mixture of Gamma distributions \cite[eq. 25]{pablo2021}. Hence, the statistics of the MFTR distribution when conditioned on $\theta$ are those of the $\kappa$-$\mu$ shadowed model, so that the PDF and CDF can be given as:
\begin{align}
\label{Lema3eq1}
f_{\gamma\left| \theta  \right.}^{\kappa\mu s}(x)= \sum_{i=0}^{\infty}w_i\left| \theta  \right. f_X^{\rm G} \left(   \mu+i;\frac{(\mu+i)\overline{\gamma}_\theta}{\mu(1+\kappa_\theta)} ;x\right),
\end{align}
\begin{align}
\label{Lema3eq2}
F_{\gamma\left| \theta  \right.}^{\kappa\mu s}(x)= \sum_{i=0}^{\infty}w_i\left| \theta  \right. F_X^{\rm G} \left(   \mu+i;\frac{(\mu+i)\overline{\gamma}_\theta}{\mu(1+\kappa_\theta)} ;x\right),
\end{align}
where the PDF and CDF of the Gamma distribution were defined in \eqref{eq17}, 
%\begin{align}
%\label{Lema3eq3}
% f_X^{\rm G}\left(  \lambda;\nu ;y\right)=\frac{ \lambda^ \lambda}{\Gamma( \lambda) \nu^ \lambda}y^{\lambda-1}\exp\left( - \frac{\lambda y}{\nu} \right),
%\end{align}
%\begin{align}
%\label{Lema3eq4}
%F_X^{\rm G}\left(  \lambda;\nu ;y\right)=\frac{1}{\Gamma(\lambda)}\gamma\left ( \lambda,\frac{\lambda y}{\nu} \right ).
%\end{align}
and the conditional mixture coefficients are given by
\begin{align}
\label{Lema3eq5}
w_i\left| \theta  \right.=\frac{\Gamma(m+i)(\mu \kappa_\theta)^i m^m}{\Gamma(m)\Gamma(i+1) \left ( \mu \kappa_\theta +m \right )^{m+i}}.
\end{align}
Applying the connection between the $\kappa$-$\mu$ shadowed and the MFTR models as shown in~\eqref{Lema1eq7} and~\eqref{Lema1eq10}, the PDF and CDF distributions of the MFTR model are formulated after averaging over the distribution of $\theta$ as 
\begin{align}
\label{Lema3eq6}
f_{\gamma}(x)= \sum_{i=0}^{\infty}w_i f_X^{\rm G} \left(   \mu+i;\frac{(\mu+i)\overline{\gamma}}{\mu(1+\kappa)} ;x\right),
\end{align}
\begin{align}
\label{Lema3eq7}
F_{\gamma}(x)= \sum_{i=0}^{\infty}w_i F_X^{\rm G} \left(   \mu+i;\frac{(\mu+i)\overline{\gamma}}{\mu(1+\kappa)} ;x\right),
\end{align}
where now the weighting coefficients, $w_i$, are rewritten as
\begin{align}
\label{Lema3eq8}
w_i=&\frac{\Gamma(m+i)(\mu K)^i m^m}{\Gamma(m)\Gamma(i+1)} \nonumber \\ &  \times\underset{I_2}{\underbrace{\frac{1}{\pi}\int_{0}^{\pi}\frac{  \left( 1+\Delta \cos \theta  \right)^i}{\left( \mu K\left( 1+\Delta \cos \theta  \right)+m \right)^{m+i} }d\theta}}.
\end{align}
{Using the results in the Appendix of \cite{ftr2}, $I_2$ can be evaluated in closed-form fashion as
\begin{align}
\label{Lema3eq9}
\begin{array}{l}
 I_2  = \frac{{\left( {1 - \Delta } \right)^i }}{{\sqrt \pi  (\mu K(1 - \Delta ) + m)^{m + i} }}\sum\limits_{q = 0}^i {\left( {\begin{array}{*{20}c}
   i  \\
   q  \\
\end{array}} \right)} \frac{{\Gamma \left( {q + \frac{1}{2}} \right)}}{{\Gamma \left( {q + 1} \right)}}\left( {\frac{{2\Delta }}{{1 - \Delta }}} \right)^q  \\ 
 \quad \quad \quad \quad \quad  \times {}_2F_1 \left( {m + i,q + \frac{1}{2};q + 1;\frac{{ - 2\mu K\Delta }}{{\mu K(1 - \Delta ) + m}}} \right), \\ 
 \end{array}
\end{align}
when $m$ is an arbitrary real number.
%and with \eqref{eq:appc:7} when $m$ is also an integer number
%\begin{align}
%\label{Lema3eq9b}
%\begin{array}{l}
% I_2  = \frac{{\left( {1 - \Delta } \right)^i }}{{\sqrt \pi  (\mu K(1 - \Delta ) + m)^{m + i} }}\sum\limits_{q = 0}^i {\left( {\begin{array}{*{20}c}
%   i  \\
 %  q  \\
%\end{array}} \right)} \frac{{\Gamma \left( {q + \frac{1}{2}} \right)}}{{\Gamma \left( {q + 1} \right)}}\left( {\frac{{2\Delta }}{{1 - \Delta }}} \right)^q  \\ 
% \quad \quad  \times \left( {\frac{{\mu K(1 - \Delta ) + m}}{{\mu K(1 + \Delta ) + m}}} \right)^{q + \frac{1}{2}}  \\ 
%  \times \sum\limits_{r = 0}^{m + i - q - 1} {( - 1)^r } \left( {\begin{array}{*{20}c}
 %  {m + i - q - 1}  \\
 %  r  \\
%\end{array}} \right)\frac{{\left( {q + \frac{1}{2}} \right)_r }}{{\left( {q + 1} \right)_r }}\left( {\frac{{2\mu K\Delta }}{{\mu K(1 + \Delta ) + m}}} \right)^r.  \\ 
% \end{array}
%\end{align}
}
Then, by substituting~\eqref{Lema3eq9} into~\eqref{Lema3eq8} and then combining the resulting expression with~\eqref{Lema3eq6} and~\eqref{Lema3eq7}, the PDF and CDF of the MFTR model for arbitrary values can be expressed as in~\eqref{eq15} and~\eqref{eq16}, respectively. This completes the proof.

{
}

%%%%%%%%%%%%%%%%%%%%%%%%%%%%%%%%
%%%%%%%%%%%%%%%%%%%%%%%%%%%%%%%%
%APPENDIX 
%\newpage
\section{Proofs of Propositions~\ref{propo1}, and~\ref{propo2}}
\label{ap:asd}

%%%%%%%%% apendice D
%%%%%%%%%%%%%%%%%%%%%%%%%%%%%%%

%%%%%%%%%%%%%%%%%%%%%%%%%%%%%%%%%%%%%%%%%%%%  Asymptotic ASC
\subsection{Asymptotic Outage Probability}\label{AsympCDF}
 In order to derive the asymptotic outage probability, we take advantage of the connections between the $\kappa$-$\mu$ shadowed and the MFTR models obtained in Appendix~\ref{ap:MGF}.
 For this purpose, we start from the conditional\footnote{In practice, it suffices to replace $\overline{\gamma}$ and $\kappa$ by $\overline{\gamma}_\theta$ and $\kappa_\theta$, respectively, in the original $\kappa$-$\mu$ shadowed asymptotic CDF, as indicated in Appendix~\ref{ap:MGF}.  }$\kappa$-$\mu$ shadowed asymptotic CDF, given by~\cite[Eq.~(35)]{Fernandez2020}
\begin{align}
\label{apenIIIeq1}
F_{\gamma\left| \theta  \right.}^{\kappa\mu s}(x)\simeq \frac{\mu^{\mu-1}(1+\kappa_\theta)^{\mu}m^m}{\left(\kappa_\theta\mu+m  \right)^m\Gamma(\mu)}\left(\frac{x}{\overline{\gamma}_\theta}  \right)^{\mu}.
\end{align}
Next, using the relationships between the $\kappa$-$\mu$ shadowed and the MFTR models given in~\eqref{Lema1eq7} and~\eqref{Lema1eq10} into~\eqref{apenIIIeq1}, we get
\begin{align}
\label{apenIIIeq2}
{F_{\gamma}(x)}\simeq& \frac{\mu^{\mu-1}(1+K)^{\mu}m^m}{\Gamma(\mu)}\left(\frac{x}{\overline{\gamma}}  \right)^{\mu} \nonumber \\ & \times \underset{I_3}{\underbrace{\frac{1}{\pi}\int_{0}^{\pi}\frac{1}{\left(\mu K\left(1+\Delta\cos\theta  \right)  +m  \right)^m}d\theta}}.
\end{align}
Employing~\cite[eq.~(38)]{Corrales2019}, the integral in $I_3$ can be expressed in simple exact closed-form as   
\begin{align}
\label{apenIIIeq3}
I_3=&\left( m-\left( \Delta-1 \right)K \mu\right)^{-m}\nonumber \\ & \times {}_{2}F_1 \left( \frac{1}{2},m;1;\frac{2\Delta K \mu}{K\mu\left( 1-\Delta \right)+m} \right)
\end{align}
Substituting~\eqref{apenIIIeq3} into~\eqref{apenIIIeq2}, the MFTR asymptotic CDF is attained as
\begin{align}
\label{apenIIIeq4}
F_{\gamma}(x)\simeq& \frac{\mu^{\mu-1}(1+K)^{\mu}m^m}{\left( m-\left( \Delta-1 \right)K \mu\right)^m\Gamma(\mu)}\left(\frac{\gamma}{\overline{\gamma}}  \right)^{\mu} \nonumber \\ & \times {}_{2}F_1 \left( \frac{1}{2},m;1;\frac{2\Delta K \mu}{K\mu\left( 1-\Delta \right)+m} \right).
\end{align}
Finally, taking into account that $P_{\rm out}\simeq F_{\gamma}(2^{R_{th}}-1)$ with $F_{\gamma}(\cdot)$ given in~\eqref{apenIIIeq4}, the asymptotic outage probability can be reached as in~\eqref{eq25}, which concludes the proof.

%%%%%%%%%%%%%%%%%%%%%%%%%%%%%%%%%%%%%%%%%%%%  Asymptotic ASC
\subsection{\rm AoF}\label{AoF}
Similar to the methodology described in Appendix~\ref{AsympCDF}, we calculate the AoF of the MFTR model from the conditional second moment of the $\kappa$-$\mu$ shadowed distribution, which is given by~\cite[Eq.~(3.10)]{celiathesis}
\begin{align}
\label{apenIIIeq5}
\mathbb{E}\left \{ \gamma_{\left| \theta  \right.}^2 \right \}^{\kappa\mu s}=\frac{\overline{\gamma}_\theta^2 \left( m(1+\mu)(1+2 \kappa_\theta)+\mu \kappa_\theta^2(1+m) \right) }{m \mu (1+\kappa_\theta)^2}
\end{align}
Then, using the connections between the $\kappa$-$\mu$ shadowed and the MFTR models obtained in~\eqref{Lema1eq7} and~\eqref{Lema1eq10} into~\eqref{apenIIIeq5}, the second moment of the MFTR model is expressed as
\begin{align}
\label{apenIIIeq6}
\mathbb{E}\left \{ \gamma^2 \right \}=\frac{\overline{\gamma}^2 \left( m(1+\mu)(1+2 \mathbb{E}\left \{ \kappa_\theta \right \})+\mu \mathbb{E}\left \{ \kappa_\theta^2 \right \}(1+m) \right) }{m \mu (1+K)^2}
\end{align}
where
\begin{align}
\label{apenIIIeq7}
\mathbb{E}\left \{ \kappa_\theta \right \}=\frac{1}{\pi}\int_{0}^{\pi}K(1+\Delta \cos \theta)d\theta=K,
\end{align}
and
\begin{align}
\label{apenIIIeq8}
\mathbb{E}\left \{ \kappa_\theta^2 \right \}=\frac{1}{\pi}\int_{0}^{\pi}\left( K(1+\Delta \cos \theta) \right)^2d\theta=K^2\left( 1+\frac{\Delta^2}{2} \right).
\end{align}
Substituting \eqref{apenIIIeq7} and~\eqref{apenIIIeq8} into~\eqref{apenIIIeq6} yields
\begin{align}
\label{apenIIIeq9}
\mathbb{E}\left \{ \gamma^2 \right \}=&\frac{\overline{\gamma}^2}{m \mu (1+K)^2}\Biggl(   m(1+\mu)(1+2 K)+\mu K^2 \nonumber \\ & \times
\left( 1+\frac{\Delta^2}{2} \right) (1+m)      \Biggl).
\end{align}
Finally, inserting~\eqref{apenIIIeq9} together with $\mathbb{E}\left \{ \gamma \right \}^2=\overline{\gamma}^2$ into~\eqref{eq26} and after some algebraic manipulations, the AoF of the MFTR model can be attained as in~\eqref{eq27}. This completes the proof.

%%%%%%%%%%%%%%%%%%%%%%%%%%%%%%%%%%%%%%%%%%%%%%%%%%%%%%%%%%%%%%%%%%%%%%%%%%%%%%%%
%%%%%%%%%%%%%%%%%%%%%%%%%%%%%%%%
%BIBLIOGRAPHY

\bibliographystyle{ieeetr}
\bibliography{bibfile}

\end{document}